\documentclass[acmsmall, screen, nonacm]{acmart}

\AtBeginDocument{%
  }

\setcopyright{acmlicensed}
\copyrightyear{2025}
\acmYear{2025}
\acmDOI{XXXXXXX.XXXXXXX}

\usepackage{amsmath}
\usepackage{amsfonts}

\usepackage{algorithm}
\usepackage{algpseudocodex}
\AtBeginEnvironment{algorithmic}{\small}
\usepackage{listings}
\usepackage{booktabs}
\usepackage[capitalize]{cleveref}
\usepackage[disable]{todonotes}
\usepackage{titlesec}
\usepackage{xspace}

\DeclareMathOperator*{\argmin}{argmin}

\newcommand{\rbnote}[1]{}
\newcommand{\bmnote}[1]{}

\newcommand{\antiunify}{\textsc{Anti-unify}}
\newcommand{\cntr}{c}               %
\newcommand{\coax}{\langle\!\langle s\rangle\!\rangle}
\newcommand{\dom}{\operatorname{dom}}
\newcommand{\expr}{\varepsilon}   %
\newcommand{\exprmap}{\lambda}
\newcommand{\exprs}{\mathcal{E}}   %
\newcommand{\id}{\operatorname{id}}
\newcommand{\ite}{\operatorname{ite}}
\newcommand{\variables}{\mathbb{V}} %
\renewcommand{\state}{x} %
\newcommand{\states}{X}
\newcommand{\statevars}{\mathbb{X}}

\newcommand{\action}{a}
\newcommand{\actions}{A}
\newcommand{\game}{\mathcal{G}}
\newcommand{\Holes}{\mathcal{H}}
\newcommand{\inputs}{I}
\newcommand{\inputval}{\ensuremath{\iota}}
\newcommand{\inputvars}{\mathbb{I}}
\newcommand{\invpred}{\Phi_{\invariant}}
\newcommand{\newp}{\operatorname{new}}
\newcommand{\params}{P}
\newcommand{\paramval}{p}
\newcommand{\paramvars}{\mathbb{P}}

\newcommand{\init}{init}
\newcommand{\envtrans}{\varsigma}

\newcommand{\nrass}{m}
\newcommand{\nrgua}{n}
\newcommand{\rankpred}{\Phi_{\invariant, (l,r)}}

\newcommand{\systrans}{\sigma}
\newcommand{\vars}{\operatorname{vars}}

\newcommand{\toolname}{Prasanva\xspace}

\newcommand{\wincond}{\phi}
\newcommand{\justiceguarantee}{\textit{Js}}
\newcommand{\justiceassumption}{\textit{Je}}
\newcommand{\false}{\ensuremath{\operatorname{False}}}
\newcommand{\true}{\ensuremath{\operatorname{True}}}

\newcommand{\var}[1]{\ensuremath{\mathbf{#1}}}

\renewcommand{\wp}{\operatorname{wp}}

\newcommand{\always}{\Box}
\newcommand{\eventually}{\Diamond} %
\newcommand{\infoften}{\always\eventually}

\newcommand{\invariant}{Inv}

\newcommand{\children}{\operatorname{children}}
\newcommand{\E}{E}    %
\newcommand{\synt}{\operatorname{synt}}
\newcommand{\F}{\mathcal{F}}
\newcommand{\leaves}{\operatorname{leaves}}
\newcommand{\nodes}{\operatorname{nodes}}
\newcommand{\roots}{\operatorname{roots}}
\newcommand{\tree}{\operatorname{tree}}

\newcommand{\clustersimilar}{\textsc{Clustersimilar}\xspace}
\newcommand{\generalizedissimilar}{\textsc{GeneralizeDissimilar}\xspace}
\newcommand{\generalizesimilar}{\textsc{GeneralizeSimilar}\xspace}
\newcommand{\merge}{\textsc{Merge}\xspace}

\AtEndPreamble{
    \newtheorem{assumption}[theorem]{Assumption}
}
\AtBeginDocument{%
}
\titleformat*{\paragraph}{\bfseries}

\begin{document}

\title{Parameterized Infinite-State Reactive Synthesis}

\author{Benedikt Maderbacher}
\email{benedikt.maderbacher@tugraz.at}
\author{Roderick Bloem}
\email{roderick.bloem@tugraz.at}
\affiliation{%
  \institution{Graz University of Technology}
  \city{Graz}
  \country{Austria}
 }

\renewcommand{\shortauthors}{B. Maderbacher and R. Bloem}



\thanks{This project has received funding from the Austrian research promotion agency FFG under project FATE (\textnumero~894789).}

\begin{abstract}

We propose a method to synthesize a parameterized infinite-state systems that can be instantiated for different parameter values.
The specification is given in a parameterized temporal logic that allows for data variables as well as parameter variables that encode properties of the environment.
Our synthesis method runs in a counterexample-guided loop consisting of four main steps: First, we use existing techniques to synthesize concrete systems for some small parameter instantiations. Second, we generalize the concrete systems into a parameterized program. Third, we create a proof candidate consisting of an invariant and a ranking function. Fourth, we check the proof candidate for consistency with the program.
If the proof succeeds, the parameterized program is valid. 
Otherwise, we identify a parameter value for which the proof fails and add a new concrete instance to step one.
To generalize programs and create proof candidates, we use a combination of anti-unification and syntax-guided synthesis to express syntactic differences between programs as functions of the parameters.

We evaluate our approach on examples from the literature that have been extended with parameters as well as new problems.
\end{abstract}

\maketitle              %

\section{Introduction}

Reactive synthesis on infinite-state domains like integers or reals can be  used in settings such as cyber-physical systems that operate on complex data.
These systems are often deployed in different environments. We propose to encode the differences in the environment as parameters. We want to construct a single \emph{parameterized} system that can be used in all environments.
By synthesizing a parameterized system, we eliminate the need to create an infinite number of instances, one for each parameter value. Unlike classical reactive synthesis, where the system's size is heavily influenced by the specification, our approach constructs systems whose size is independent of the parameter values.

\begin{example} \label{ex:runningintro}
A robot with location $x \in \mathbb{Z}$ can move on a track between a loading zone (where $x = min$) and an unloading zone ($x = max$). It is affected by wind disturbances that can push it $d$ steps  in either direction.
The robot must visit  the loading and the unloading zone infinitely often.
We want a single parameterized controller that works for any value of the parameters $min$ and $max$ with $min + 1 < max$. The specification written using the linear temporal operators $\always$ for globally and $\eventually$ for eventually is included in \cref{lst:spec}.

It contains  a set of actions that  consist of a guard that states when the action may be executed and an update that defines how the state variables are changed. A globally guarantee is syntactic sugar for a condition that is added to every guard.
In this example there are three \textbf{actions}:
    \begin{itemize}
        \item $(min \leq x \leq max,\; x' = x + d)$,
        \item $(min \leq x \leq max,\; x' = x + 1 + d)$, and
        \item $(min \leq x \leq max,\; x' = x - 1 + d)$.
    \end{itemize}

The robot can stop, move right, or move left, while  wind disturbance $d$  influences the position.
The guards specify that actions can only be executed if the robot is on the track.

The \textbf{specification} $\phi$ consists of the environment transition relation $-1 \leq d \leq 1$ and the temporal specification $(\infoften d<0 \land \infoften 0<d) \Rightarrow (\infoften x = min \land \infoften x = max)$, which includes the assumption that the wind does not always blow from the same direction, as this would make it impossible to reach one of the goals.
The system must work for all initial states $min \leq x \leq max$.

The  main loop of a parameterized program that satisfies $\phi$ is shown in \cref{lst:examplesystem}.
It is executed in an infinite loop, once in  every time step.
The program uses an additional variable $\cntr$ to track the goal that is currently being targeted.
The first part of the program changes the goal if the current goal has been reached.
The program has a subprogram for each goal.
If $\cntr = 1$, the program tries to decrease $x$ until it becomes equal to $min$. The subprogram may fail to decrease $x$ if there is a head wind and must be careful not to overshoot $min$ if there is a tail wind. The case of $\cntr = 2$, where the program tries to increase $x$, is analogous. As intended, this program fulfills the specification, regardless of which concrete values are chosen for $min$ and $max$.
\end{example}

\begin{figure}[t]
\begin{minipage}{.50\textwidth}
\begin{lstlisting}[frame=tlrb, language=C, caption=Specification, label=lst:spec, mathescape=true]
Parameters: $min, max \in \mathbb{Z}$
State: $x \in \mathbb{Z}$
Input: $d \in \mathbb{Z}$
Actions:
    $\true$,  $x$ := $x - 1 + d$
    $\true$,  $x$ := $x + d$
    $\true$,  $x$ := $x + 1 + d$
assume
    $\always\; -1 \leq d \leq 1 \;\land$
    $\infoften\; d<0 \land \infoften\; 0<d$
guarantee
    $\always\; min \leq x \leq max \;\land$
    $\infoften\; x = min \land \infoften\; x = max$
\end{lstlisting}
\end{minipage}\hfill
\begin{minipage}{.45\textwidth}
\begin{lstlisting}[frame=tlrb, language=C, caption=System loop body, label=lst:examplesystem, mathescape=true]
if $(c = 1 \land x = min) \;\lor$ 
   $(c = 2 \land x = max)$:
    $c$ := $c \oplus 1$
if $c = 1$:
    if $x + d \geq min + 1$:
        $x$ := $x - 1 + d$
    else:
        $x$ := $x + d$ 
if $c = 2$:
    if $x + d \leq max - 1$:
        $x$ := $x + 1 + d$
    else:
        $x$ := $x + d$
\end{lstlisting}
\end{minipage}
\end{figure}

\paragraph{Summary} We propose a novel method to synthesize an infinite-state parameterized system that fulfills a parameterized Generalized Reactivity(1) or GR(1) specification (an extension of propositional GR(1)\cite{DBLP:journals/jcss/BloemJPPS12}). 
The method consists of a counterexample-driven synthesis loop  (\cref{fig:overview}). At  every iteration of the loop, we have a finite set of parameter values $P'$. We construct a concrete synthesis problem for each $p \in P'$, and pass them to an existing concrete synthesis method, which returns a concrete system and an invariant for each $p$. 
We use the concrete solutions as input to a generalization procedure. This procedure constructs a generalized program and a generalized candidate invariant that may both include parameter variables. We also construct a candidate ranking function for the generalized program that we use to prove progress towards the goals of the GR(1) specification.

By construction, the generalized program and invariant encompass the behaviors of all  concrete programs. 
We verify whether the generalized system is consistent with the generalized candidate invariant and ranking function, which means that the system is correct for all parameter values.
If not, we obtain a new parameter value $p$ for which the program does not match the proof candidate and we add $p$ to $P'$ and start the next iteration. 

There are several alternative executions. First, the concrete synthesis algorithm may find a counterexample to realizability for $p$, which means that the parameterized specification is  unrealizable. Second, the proof candidate may succeed in showing that the generalized program is correct, and synthesis succeeds. Third, because parameterized synthesis is undecidable, the loop may not terminate. Note that the concrete synthesis problems may or may not be be finite state. For instance, the parameters may bound the range of an integer, which means that any concrete system is finite state. If, instead, the parameters bound a real-valued variable, a concrete system is infinite. In the latter case, the concrete synthesis problem is also undecidable and may not terminate. 

The parameterized synthesis problem is equivalent to infinite state synthesis: Parameterized synthesis can be reduced to infinite state synthesis by defining the parameter as an additional input that does not change during execution, and infinite state synthesis is simply parameterized synthesis using an empty set of parameters. However, existing infinite state synthesis methods are based on Boolean abstraction or symbolic fixpoint computations, rather than generalization. We will show in the experimental results that our approach performs much better on parameterized instances.

We present the following contributions.%
\begin{itemize}
    \item A formulation of the parameterized synthesis problem for GR(1) specifications.
    \item A proof that parameterized synthesis is undecidable, even if the synthesis problem for concrete instances is decidable.
    \item A counterexample-driven synthesis loop for parameterized specifications.
    \item A generalization procedure based on anti-unification and syntax-guided synthesis for programs and invariants.
    \item A method to construct ranking functions for parameterized programs.
    \item An algorithm for (the decidable problem) of checking whether a generalized program is consistent with a proof candidate.
    \item A prototype implementation called \toolname and an experimental evaluation.
\end{itemize}
\begin{figure}
    \centering
    \includegraphics[width=\linewidth]{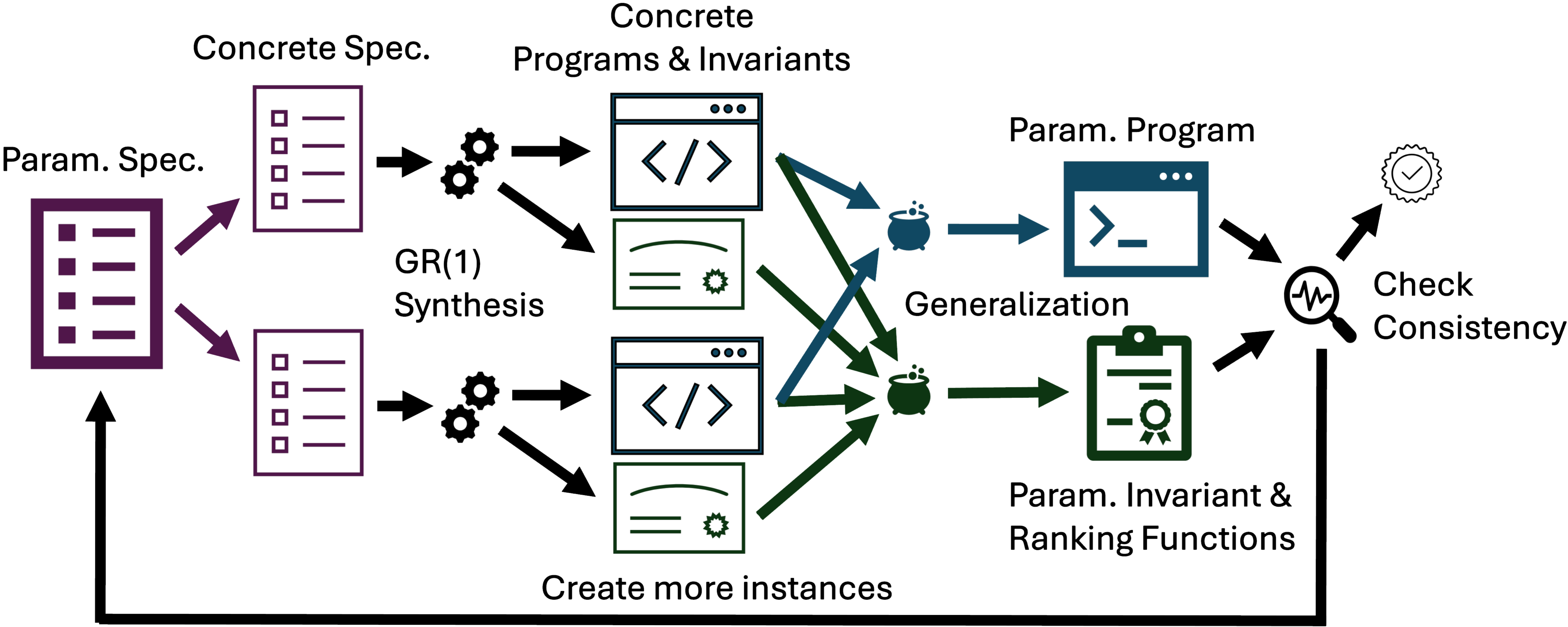}
    \caption{Overview of the counterexample-driven synthesis loop.}
    \label{fig:overview}
\end{figure}

\section{Problem Statement}
\label{sec:problem_statement}
We will now define our problem formally. A parameterized synthesis problem has a finite set $\variables$ of variables partitioned into parameter variables $\paramvars$, state variables $\statevars$, and input variables $\inputvars$. All state variables $\statevars$ are observable by the environment and we thus do not define an additional set of output variables.

Parameter variables are always integer. State and input variables can be Boolean, integer, or real.
When the distinction is relevant we use $\statevars_{\mathbb{B}}$/$\inputvars_{\mathbb{B}}$, $\statevars_{\mathbb{Z}}$/$\inputvars_{\mathbb{Z}}$, and $\statevars_{\mathbb{R}}$/$\inputvars_{\mathbb{R}}$ for Boolean, integer, and real state/input variables respectively. 
\begin{assumption}\label{ass:decidablelogic}
We assume that $\paramvars$ is a set of integer variables and that all expressions in $\exprs(\variables)$ are in linear integer arithmetic, real arithmetic, or mixed real and linear integer arithmetic.
\end{assumption}
This assumption ensures that satisfiability is decidable and that there are effective means to perform quantifier elimination.

Let $V$, $\params$, $\states$, and $\inputs$ be the sets of all valuations of $\variables$, $\paramvars$, $\statevars$, and $\inputvars$, respectively.
A valuation $x \in X$ is  a \emph{state}. 
Systems and specifications are represented symbolically.
To denote the values of variables in the next state we  use a set $\statevars'$ of primed copies of $\statevars$.

We will write $\exprs(\variables)$ for the set of expressions from the theory over $\variables$. An \emph{update} $u$ assigns an expression $u(\var{x}) \in \exprs(\variables)$ to every variable $\var{x} \in \statevars$. 
We represent an update $u$ symbolically as $\mathbf{u} =  \bigwedge_{\var{x} \in \statevars} (\var{x}' = u(\var{x}))$.
A \emph{guard} is a predicate over $\variables$ and 
an \emph{action} $a$ is a pair $(g,u)$, where $g$ is a guard and $u$ is an update.
We will represent an action $a = (g,u)$ symbolically as $\mathbf{a} = g \wedge \mathbf{u}$. Note that  guards and updates may include parameters and that for some parameter instantiations, there may not be any initial states.%

A \emph{parameterized system} is a system that can be instantiated for different parameter values from $P$ to yield a concrete system. Mathematically, it is a tuple $M = (\paramvars,\statevars,\inputvars,\init, \actions)$, where %
$\init \in \exprs(\statevars \cup \paramvars)$ is a predicate that defines the initial states and $\actions$ is a set of actions.
Given a parameter instantiation $p \in P$, a \emph{concrete system} $M_p = (\statevars, \inputvars, \init_p, \actions_p)$ is obtained by substituting each variable  $\mathbf{p} \in \paramvars$ with  $p(\mathbf{p})$. %
Let $\systrans_p = \bigvee_{a \in \actions_p} \mathbf{a}$ be the transition relation formed by the actions in $\actions_p$.
A sequence $\pi = x_0, \inputval_0, x_1, \inputval_1, \dots$ with states $x_t \in X$ and inputs $\inputval_t \in I$ is a \emph{trace} of $M_p$ if $x_0 \models \init_p$ and for all $t$, we have $x_t \cup  \inputval_t \cup x'_{t+1} \models \systrans_p$. 

A parameterized infinite-state GR(1) specification is a tuple $\game = (\paramvars,\statevars,\inputvars, \init, \envtrans, \actions, \wincond)$, where 
$\init$, and
$\actions$ are as before; %
$\envtrans$ is the environment transition relation, a predicate over $\paramvars \cup \statevars \cup \inputvars$,
and $\wincond$ is a GR(1) winning condition. The winning condition has the form \[\bigwedge_{i=1}^{\nrass} \infoften \justiceassumption_i \rightarrow \bigwedge_{j=1}^{\nrgua} \infoften \justiceguarantee_j,\] where $\justiceassumption_i$ and $\justiceguarantee_j$ are sets of $m$ and $n$ predicates over $\statevars \cup \paramvars$, called the \emph{justice assumptions} and \emph{guarantees}, resp.
A concrete specification $\game_p  = (\statevars,\inputvars, \init, \envtrans, \actions, \wincond)$ is obtained by substituting $\paramvars$ according to $p$. 

The \emph{parameterized synthesis problem} for a parameterized  specification $\game = (\paramvars, \statevars,\inputvars,\init, \envtrans, \actions, \wincond)$ (and possibly a subset $P' \subset P$), is to find a parameterized system $M = (\paramvars,\hat{\statevars},\inputvars, \widehat{\init}, \hat{\actions})$ such that 
\begin{enumerate}
    \item $\hat{\statevars} \supseteq \statevars$, $\widehat{\init} \rightarrow \init$, and for all $xp \in X \times P$ where $xp\models\init$ there exists $\hat{x}p \in \hat{X}\times P$ such that $\hat{x}p \downharpoonright X \times P = xp$ and $\hat{x}p \models \widehat{\init}$. (Extra state variables are allowed, but the initial states are otherwise unchanged.)
    \item  $M$ changes $\statevars$ according to the actions $\actions$: for all $(\hat{g},\hat{u}) \in \hat{\actions}$ there is an action $(g, u) \in \actions$ such that $\hat{g} \rightarrow g$ and for all $\var{x} \in \statevars$, $\hat{u}(\var{x}) = u(\var{x})$.
    \item $M$ is \emph{correct (for $P'$)}: for all $p \in P$ ($p\in P'$, resp.), we have  $M_p \models \game_p$.   (See below for the semantics of GR(1).)
\end{enumerate}

\section{Preliminaries}
We write %
$\vars(\expr)$ for the set of variables present in an expression $\expr$ and $\mathcal{S}(V): V \rightarrow \exprs(V)$ (or just $\mathcal{S}$) for the set of substitutions of variables in $V$. For $s \in \mathcal{S}$, we write $\expr[s]$ for the application of $s$ on $\expr$, and $v \gets \expr$ for the substitution  that replaces just $v$ with $\expr$. We will use $a \oplus_n b$ to denote addition modulo $n$, and omit $n$ when obvious.

\subsection{Generalized Reactivity(1)}
\label{sec:gr1}

We will use the linear temporal operators $\always$ for globally and $\eventually$ for eventually. The formula $\infoften \phi$ means that $\phi$ happens infinitely often. 

Generalized Reactivity(1) or GR(1) was introduced in \cite{DBLP:conf/vmcai/PitermanPS06,DBLP:journals/jcss/BloemJPPS12}. As briefly mentioned above, the defining aspect of GR(1) specifications is the winning condition, which is of the form  
\begin{equation}
    (\infoften \justiceassumption_1 \wedge \dots \wedge \infoften \justiceassumption_{\nrass}) \rightarrow (\infoften \justiceguarantee_1 \wedge \dots \wedge \infoften \justiceguarantee_{\nrgua}).\label{eq:gr1-winning}
\end{equation}
The GR(1) synthesis algorithm can be implemented as a triply nested fixpoint, only uses a quadratic number of preimage computations.  Even though preimage computations are not constant time, this complexity compares favorably to more general specification languages like Linear Temporal Logic, which are 2EXP-hard \cite{DBLP:conf/popl/PnueliR89}.

Because the specification language allows for very efficient implementation, while remaining expressive enough for realistic specifications, GR(1) synthesis has been implemented in several tools \cite{DBLP:journals/jcss/BloemJPPS12,cav/EhlersR16,DBLP:journals/sosym/MaozR21,Livingston2024gr1c,Livingston24gr1py,gagliardi2025scalinggr1synthesiscompositional} and applied on a series of case studies in different application domains \cite{DBLP:journals/entcs/BloemGJPPW07,Godhal13,DBLP:journals/corr/MaozR16,DBLP:conf/amcc/ZehfrooshT19,Zhu22,DBLP:conf/models/WeteGWFK21,gagliardi2025scalinggr1synthesiscompositional}.

A GR(1) specification states that if the environment satisfies each condition $\justiceassumption_i$ infinitely often, then the system satisfies its conditions $\justiceguarantee_j$ infinitely often as well. It is important to note that the ``generalized'' in generalized reactivity refers to a  different notion than the generalization that we use in this paper: GR(1) specifications are closely related to Reactivity (Streett) specifications \cite{DBLP:journals/iandc/Streett82,DBLP:books/daglib/0077033}. They simultaneously restrict these specifications to only one pair (hence the ``(1)'') and generalize them by requiring multiple conditions to be met infinitely often.

Formally, a GR(1) specification is a concrete specification $\game = (\statevars,\inputvars,\init, \envtrans, \actions, \wincond)$, where the winning condition $\wincond$ is of the form shown in \cref{eq:gr1-winning}. Writing $\systrans = \bigvee_{a \in A}\mathbf{a}$ as before, we have the following definition of the semantics.

\begin{definition}\label{def:gr1semantics}
    a system $M$ fulfills a GR(1) specification $\game$, denoted by $M \models \game$, if two conditions hold. 
    
    (1) There are no deadlocks, i.e., for all finite traces $\pi = x_0, \inputval_0,  \dots, x_n, \inputval_{n}$ of $M$ such that for all $i$ we have $x_i, \inputval_{i+1},x_{i+1} \models \envtrans$ (the environment transition relation is fulfilled), there is an  $x_{n+1}$ such that $x_n, \inputval_n, x_{n+1} \models \sigma$.

    (2) for all traces  $\pi = x_0, \inputval_0, x_1, \inputval_1, \dots$ of $M$ such that for all $i$ we have $x_i, \inputval_{i+1},x_{i+1} \models \envtrans$, we have (a) for each $i$, $x_i, \inputval_{i+1},x_{i+1} \models \systrans$ (the system transition relation is fulfilled) and (b) $\pi \models \phi$ (the fairness condition is fulfilled).
\end{definition}

Note that a restrictions $\phi$ on the states that the system may visit can be encoded by conjoining every guard with $\phi$, so that undesired executions will deadlock.

We see the specification is seen as a \emph{game} between an environment player that controls the input $\inputs$ and a system player that controls the state $\states$. The question is whether (and from which states) a \emph{strategy} exists so that no matter which inputs the environment player chooses,  the system player can always choose a transitions so that the resulting infinite trace satisfies the specification. 

As a building block to compute this set of states, we need the coax operator $\coax$, where for $\states' \subseteq \states$  we denote by $\coax \states'$ the set of states from which the system player can  the system to some $\state \in \states'$. Recalling that $\systrans = \bigvee_{a \in \actions}$, we have
$$
\coax \states' = \{ \state \in \states \mid 
        \forall \inputval\in\inputs.\,  \state, \inputval \models \envtrans \rightarrow
            \exists \state' \in \states.\, \state, \inputval, \state' \models \systrans \text{ and } \state' \in \states'
        \}. 
$$

\paragraph{Finite State GR(1) Synthesis.}
The classical algorithm for  GR(1) synthesis with Boolean variables \cite{DBLP:journals/jcss/BloemJPPS12} decides whether a concrete GR(1) specification is \emph{realizable}, that is, whether a system exists that satisfies it. If so, it returns a system $M = (\statevars \cup \{\cntr\}, \inputvars, \widehat{\init}, \hat{\actions})$ and a predicate $\invariant$ over $\statevars$ representing the winning states. 

The GR(1) algorithm computes the following triply nested fixpoint.
$$
\invariant = \nu
\begin{bmatrix}
Z_1 \\
Z_2 \\
\vdots \\
Z_{\nrgua}
\end{bmatrix} .\, 
\left[
\begin{array}{c}
\mu Y.\, \left( \bigvee_{i=1}^{\nrass} \nu X.\, \left( J^s_1 \land \coax Z_2 \lor \coax Y \lor \neg J^e_i \land \coax X \right) \right) \\
\mu Y.\, \left( \bigvee_{i=1}^{\nrass} \nu X.\, \left( J^s_2 \land \coax Z_3 \lor \coax Y \lor \neg J^e_i \land \coax X \right) \right) \\
\vdots \\
\mu Y.\, \left( \bigvee_{i=1}^{\nrass} \nu X.\, \left( J^s_n \land \coax Z_1 \lor \coax Y \lor \neg J^e_i \land \coax X \right) \right)
\end{array}
\right].
$$
The inner greatest fixpoints (with iterates $X$) find cycles of states in which some assumption $i$ is never visited. The middle  least fixpoints (with iterates $Y$) find paths that either go into such a cycle or eventually reach $\justiceguarantee_j$. Finally, the outer greatest fixpoint (with iterate $Z$) finds a cycle that either repeatedly visits all guarantees or ever visits some assumption $\justiceassumption_i$. A naive implementation of the triply nested fixpoint gives a cubic algorithm. However, the formulas $\bigvee_{i=1}^{\nrass} \nu X.\, ( J^s_j \land \coax Z_{j\oplus 1} \lor \coax Y \lor \neg J^e_i \land \coax X )$ are monotonic in $Z_{j\oplus 1}$. Thus, in each iteration of the outer fixpoint for $\invariant$, we can start the computation for $X$ with the previous fixpoint, yielding a quadratic complexity. (See \cite{DBLP:journals/tcs/BrowneCJLM97}.)

\begin{algorithm}[t]
\begin{algorithmic}[1]
\State $\invariant \gets \true$
\While{$\invariant$ changed}
    \For{$j \in [1,n]$}
        \State $Y \gets \false$
        \State $r \gets 0$
        \While{$Y$ changed}
            \State $start \gets (\justiceguarantee_j \land \coax \invariant) \lor \coax Y$
            \State $Y \gets \false$
            \For{$i \in [1,m]$}
                \State $X \gets \invariant$
                \While{$X$ changed}
                    \State $X \gets start \lor (\neg \justiceassumption_i \land \coax X)$
                \EndWhile
                \State $Y \gets Y \lor X$
                \State $mX[j][r][i] \gets X$
            \EndFor
            \State $mY[j][r] \gets Y$
            \State $r \gets r + 1$
        \EndWhile
        \State $\invariant \gets Y$
    \EndFor
\EndWhile
\State \Return{$\invariant$}
\end{algorithmic}
\caption{GR(1) fixpoint algorithm} 
\label{alg:gr1fixpoint}
\end{algorithm}

\Cref{alg:gr1fixpoint} shows the implementation of the triply nested fixpoint that returns the set of states $\invariant$ from which a winning strategy exists. In case of success, that is, if the initial state is in $\invariant$, we can construct the strategy using the iterates $mX$ and $mY$, as follows.

The system has memory: it has one substrategy for each of the $n$ guarantees, which it fulfills in a round-robin manner. A counter $\cntr \in [1,\nrgua]$  keeps track of the next guarantee to be fulfilled. By convention, it starts with the first guarantee and thus has initial state $\widehat{\init} = \init \wedge c = 1$. %

The actions are partitioned into sets  $\{\actions^j\}_{j \in [1,\nrgua]}$, so that the transition relation is $\systrans = \bigvee_{j=1}^{\nrgua} \systrans^j$, where $\systrans^j$ is defined by $\actions^j$. For every $j \in [1,\nrgua]$,
substrategy $j$ determines what to do when $\cntr = j$: Actions in $\hat{\actions}^j$ are of the form  $(\cntr = j) \wedge g \wedge \mathbf{u}$ for some guard $g$ and update $u$.
Furthermore, the system 
moves to the next substrategy when the current one is fulfilled: For each $x \cup \inputval \cup x' \models \systrans^j$, $x' \in \justiceguarantee_j$ iff  $x' \models (\cntr' = j \oplus 1)$.

Finally, the system
makes progress: for any trace $\pi = x_0,\inputval_0, \dots$ according to $\systrans^j$ so that $\envtrans$ is always satisfied, either $\pi$ is finite and the last state has $\cntr = j \oplus 1$, or there is a $\justiceassumption_i$ for which there are only finitely many $t$ such that $x_t \models \justiceassumption_i$.

\newcommand{\maxr}{\operatorname{maxr}}
When $\cntr = j$, the algorithm guarantees progress to guarantee $\justiceguarantee_j$ as follows. The iterates $mY[j][r]$ ($r \in [0,\maxr_j]$) give a ranking function on the distance $r$ to $\justiceguarantee_j$ with $mY[j][\maxr_j] = \invariant$. If $r=0$, the counter is increased. For $r>1$, the iterate $mY[j][r]$ is further partitioned into $mX[j][r][i]$. What happens in $mX[j][r][i]$ depends on whether $\justiceassumption_i$ is satisfied. If it is, the system can make progress to a $mY[j][r']$ with $r' <r$. If it is not, the system can stay in $mX[j][r][i]$ or move to an $mX[j][r'][i']$ with either $r'<r$ or $r'=r$ and $i'<i$. Thus, progress can be made whenever  $\justiceassumption_i$ is satisfied and if all assumptions are satisfies infinitely often, $\justiceguarantee_j$ is eventually reached.

\paragraph{Infinite State GR(1) Synthesis.}
We use the infinite state GR(1) synthesis algorithm as a building block for  parameterized synthesis. Let  $\game = (\statevars,\inputvars, \init, \envtrans, \actions, \wincond)$ be a concrete GR(1) specification and let $M = (\hat{\statevars}, \inputvars, \init, \hat{\actions})$ be a concrete system with $\hat{\statevars} \supseteq \statevars$.
We say that $M$ \emph{satisfies} $\game$ ($M \models \game$) if for any trace $x_0,\inputval_0, \dots$ of $M$,
(1) The system never deadlocks. That is, for all $t$ such that $x_t,\inputval_t \models \envtrans$  there is at least one action $(g,u) \in \hat{\actions}$ such that $x_k \cup \inputval_k \models g$;
(2) if for each $t$, $x_t \cup \inputval_t \models \envtrans$ and for all $i$ there are infinitely many $t$ such that $x_t \models \justiceassumption_i$, then for all $j$ there are infinitely many $t$ such that $x_t \models \justiceguarantee_j$.

An approach to synthesize GR(1) specifications with theory variables can be found in \cite{DBLP:conf/isola/MaderbacherWB24}. In the following, we will refer to that approach under the name \textsc{GR(1)-Synthesis}. It closely follows the fixpoint approach for Boolean specifications that we describe above by representing potentially infinite sets of states as SMT formulas. The coax operator $\coax$ is implemented using a quantifier elimination tactic on these formulas. The fixpoint algorithm (\ref{alg:gr1fixpoint}) is otherwise unchanged from the Boolean case, but crucially, GR(1) synthesis for non-Boolean variables is undecidable and thus the procedure may not terminate.  

If the fixpoint algorithm terminates, we can use the intermediate results $mX$ and $mY$ to build a system $M$.
We will use the shorthands 
\begin{align*}
mY[j][=\!r]      &= mY[j][r] \land \neg mY[j][<r],\\
mY[j][<\!r]      &= \bigvee_{r'<r} mY[j][r'],\\
mX[j][=\!(r,i)] &= mX[j][r][i] \land \neg \bigvee_{i'<i} mX[j][r][i'] \land \neg \bigvee_{r'<r}\bigvee_{i} mX[j][r'][i]\text{, and}\\  %
wp(S,a) &= \envtrans \rightarrow (g \land S[\{(x,u(x)) \mid x \in \statevars\}]) \;\text{where}\; a = (g,u).
\end{align*}
For each action $a \in \actions$ from the specification and for each substrategy $j$, the system has an action $\hat{\action} \in \hat{\actions}^j$ with a guard $\gamma(j,\action)$, where $\gamma$ is defined using three cases:
\begin{equation}
\gamma(j,\action)= \gamma_1(\action) \lor \gamma_2(j,\action) \lor \gamma_3(j,\action).
\label{eq:simplestrategy}
\end{equation}
The first part states that an action $\action$ has to remain in $\invariant$ when all guarantees are simultaneously satisfied. In that case there is no defined next substrategy and the system only has to ensure that it does not violate the invariant.
$$
        \gamma_1(\action) = \bigwedge_j \justiceguarantee_j \land \wp(\invariant, \action).
$$
The second part describes that an action can be used when it ensures progress towards the current target guarantee.
\begin{equation*}
    \begin{split}
        \gamma_2(j,r,\action) &= mY[j][=r] \land \wp(mY[j][<r], \action) \text{ and}\\
        \gamma_2(j,\action) &= \bigvee_{r\geq1} \gamma_2(j,r,\action).
    \end{split}
\end{equation*}
Finally, the third part states that if an assumption is violated, the system can take an action that remains in the same region. This means that when one assumption is eventually always violated the system does not need to make progress.
\begin{equation*}
    \begin{split}
        \gamma_3(j,r,i,\action) &= mX[j][= (r,i)] \land \neg \justiceassumption_i \land \wp(mX[j][r][i], \action)\text{ and}\\
        \gamma_3(j,\action) &= \bigvee_{r} \bigvee_{i} \gamma_3(j,r,i,\action).
    \end{split}
\end{equation*}
To obtain compact strategies the guards are simplified using a redundancy elimination algorithm \cite{DBLP:conf/sas/DilligDA10}.
The invariant and the negated earlier guards are used as additional context for the simplification.
For a detailed description of the simplification algorithm for GR(1) strategies with theory variables see \cite{DBLP:conf/isola/MaderbacherWB24}.

\subsection{Syntax-Guided Synthesis}
\label{sec:sygus}
The \emph{syntax-guided synthesis} (SyGuS) problem \cite{DBLP:series/natosec/AlurBDF0JKMMRSSSSTU15} is to construct a correct system given a formal specification of correctness, a grammar describing all possible systems, and a background theory. It assumes that a verification tool for systems in the given theory is available. Whereas deductive approaches like the one for GR(1) are characterized by a mechanical procedure to construct the system from the specification, syntax-guided synthesis approaches often work inductively, by iterating through various implementations and learning from each failure, until a correct implementation is found. \emph{Counterexample guided synthesis} in particular \cite{DBLP:conf/asplos/Solar-LezamaTBSS06} considers system candidates of increasing size and thus complexity, meaning that it tends to return the simplest system possible.

In the following, we assume that we are given a syntax-guided synthesis procedure  that takes a set of function symbols $f_1,\dots f_n$ and a specification $\phi$ containing these symbols and returns a set of functions $f^*_i$ such $\phi[f_i \gets f^*_i]$ is true. We denote this procedure by $\textsc{FuncSynth}((f_1,\dots,f_n), \phi)$.

Syntax-guided synthesis uses a grammar to define the search space.
We provide the grammars used for the various calls to \textsc{FuncSynth} in extended Backus–Naur form.
Non-terminals are written as capital letters.
We use $\paramvars, \statevars$ and $\inputvars$ as shorthand for the variable symbols contained in the respective set.

\subsection{Anti-Unification}
\label{sec:anti-unification}
Our generalization procedure uses the classic anti-unification algorithm for first order expressions \cite{plotkin1970, reynolds1970, DBLP:conf/ijcai/CernaK23}. %
We show the anti-unification algorithm for two expressions in \cref{alg:anti-unification}.
The algorithm takes two expressions $e_1$ and $e_2$ and computes a general expression $g$ with variables from $\Holes$ (``holes'') for the parts that differ, along with substitutions $s_1, s_2 \in \mathcal{S}(\Holes)$ such that $(e_1,e_2) = (g[s_1], g[s_2])$.
Expression $g$ is the \emph{least general generalization} of $e_1$ and $e_2$, i.e., the expression with the smallest holes that can be unified with both input expressions.

\begin{algorithm}[tb]
\caption{Anti-unification of two expressions}\label{alg:anti-unification}
\begin{algorithmic}
\Function{Anti-unify}{$e_1, e_2$}
    \State $s_1 \gets \{(h_0,e_1)\}$
    \State $s_2 \gets \{(h_0, e_2)\}$
    \State $g \gets h_0$
    \While{$g$ changed}
        \If{$\exists h,f,e_a,e_b,e_c,e_d,\dots: (h, f(e_a, e_b, \dots)) \in s1 \land (h, f(e_c, e_d, \dots)) \in s2$} %
            \State $g \gets g[h \gets f(h_a, h_b, \dots)]$ \Comment{Fresh variables $h_a, h_b, \dots$}
            \State $s_1 \gets s_1 \setminus \{(h, f(e_a, e_b, \dots))\} \cup \{ (h_{a}, e_a), (h_{b}, e_b), \dots \}$
            \State $s_2 \gets s_2 \setminus \{(h, f(e_c, e_d, \dots))\} \cup \{ (h_{a}, e_c), (h_{b}, e_d), \dots \}$
        \EndIf
        \If{$\exists h_a, h_b, e_a,e_c: h_a \neq h_b \land \{(h_a, e_a), (h_b,e_a)\} \subseteq s_1 \land \{(h_a, e_c), (h_b, e_c)\} \subseteq s_2$}%
            \State $g \gets g[h_a \gets h_b]$
            \State $s_1 \gets s_1 \setminus \{(h_a, e_a)\}$
            \State $s_2 \gets s_2 \setminus \{(h_a, e_c)\}$
        \EndIf
    \EndWhile
    \State \textbf{return} $g$, $s_1$, $s_2$
\EndFunction
\end{algorithmic}

\end{algorithm}

The algorithm keeps a \emph{generalizer} $g$ and two substitutions $s_1$ and $s_2$ so that in any iteration, $e_1 = g[s_1]$ and $e_2 = g[s_2]$. It starts with the most general $g$, which is a single variable and applies one of two rewrite rules until a fixpoint is reached.
The first rule makes the generalizer more specific by moving shared symbols into it, while the second rule eliminates redundancies by combining holes that are equivalent in both expressions.

As an example, consider the  expressions $e_1 = x \leq 1 \land y \leq 1$ and $e_2 = x \leq 2 \land y \leq 2$.
The algorithm starts with $g=h_0$ and the substitutions $s_1 = \{(h_0, x \leq 1 \land y \leq 1)\}$ and  $s_1 = \{(h_0, x \leq 2 \land y \leq 2)\}$.
In the first iteration, it applies Rule~1, moving the conjunction into the generalizer, which results in $g = h_1 \wedge h_2$, $s_1 = \{(h_1, x \leq 1),(h_2, y \leq 1)\}$ and   $s_2 = \{(h_1, x \leq 2),(h_2, y \leq 2)\}$. It then applies Rule~1 two more times for the comparisons and Rule~2 to eliminate the separate substitutions for the constants 1 and 2. Finally, considering $x$ and $y$ as nullary functions, it uses Rule~1 to include $x$ and $y$ into the generalizer. 
Finally, $\textsc{Anti-unify}(e_1,e_2)$ returns the generalizer $g = x \leq h \land y \leq h$ and substitutions $s_1 = \{(h,1)\}$ and $s_2 = \{(h,2)\}$.

The anti-unification of two expressions is unique (up to the names of the holes). This result carries over to arbitrary finite sets of expressions \cite{DBLP:conf/ijcai/CernaK23} and in the following, we will assume that \antiunify takes an arbitrary finite set of expressions and returns one generalizer and a substitution for each argument. In the following, we  call an expression \emph{ground} if it doesn't contain any holes.

\section{Parameterized Synthesis}\label{sec:method}
We use the insight that systems realizing a specification for different parameter values are often similar and that synthesizing a system for a small parameter is often fast. 
The main loop of our algorithm is shown in \cref{alg:param_synthesis}; we will describe its subroutines in the following sections.
First, we select a (small) parameter value $p_{\newp}$ that satisfies the $\init$ predicate in the specification $\game$.
The main synthesis loop proceeds as follows.
A concrete specification $\game_{p_{\newp}}$ is obtained by instantiating $\game$ with $p_{\newp}$. %
Next, the algorithm calls the  GR(1) game solving algorithm to  synthesize a system for this concrete specification. 
If the specification $\game_{p_{\newp}}$ is unrealizable this implies that the parameterized specification is unrealizable. 

At any point in the algorithm, we have a set of concrete systems $M_p$ and invariants $\invariant_p$ for each $p$  in a finite subset $\params' \subseteq \params$.
We generalize the concrete systems to a parameterized candidate system that is correct for $P'$ using a function \textsc{GeneralizeExpr} described in \cref{sec:generalization} to generalize expressions. For the guards, we consider one substrategy for a guarantee $\justiceguarantee_j$ at a time, and generalize the guards used for each individual update separately.

We  use the concrete systems and their invariants to construct a parameterized proof candidate.
The  proof candidate consists of a generalized invariant together with a ranking function $(l,r)$.
\Cref{sec:proofs} describes the properties of the proof candidate and how we use syntax-guided synthesis and interpolation to search for a general ranking functions. (As we will see, we cannot simply generalize the ranking functions of each concrete system).

Generalization guarantees that the candidate systems are correct for all $p \in \params'$. It also tends to produce short expressions that  generalize beyond the current set of parameters.
To verify whether a candidate  program is correct for all parameters, we check the generalized strategy against the  generalized proof candidate using an SMT solver.
If verification is successful, the candidate strategy is returned. Otherwise, either the strategy or the proof candidate (or both) are not general enough. In this case, the consistency check (see \cref{sec:proofs}) generates a new parameter as a counterexample, which is added to $P'$ for the next iteration.

\begin{algorithm}[t]
\caption{Synthesis for parameterized infinite state GR(1)}\label{alg:param_synthesis}
\begin{algorithmic}
\Function{ParameterizedSynthesis}{$\game = (\statevars, \inputvars, \init, \actions)$}
    \State $p_{\newp} \gets \textit{select parameter value that satisfies}\; \init$; $P' \gets \{p_{\newp}\}$
    \State $\textit{Systems} \gets \emptyset$
    \While{true}
        \State $\game_{p_{\newp}} \gets \textit{instantiate}\; \game \; \textit{for} \; p_{\newp}$
        \If{$\game_{p_{\newp}}$ is unrealizable}  \textbf{return} \textit{Unrealizable} \EndIf
        \State $M_{p_{\newp}}, \invariant_{p_{\newp}} \gets \textsc{GR(1)-Synthesis}(\game_{p_{\newp}})$
        
        \State $\textit{Systems} \gets \textit{Systems} \cup \{(p_{\newp}, M_{p_{\newp}}, \invariant_{p_{\newp}})\}$
        
        \State $\invariant \gets \textsc{GeneralizeExpr}(\{ (p, \invariant_p) \mid (p,\_,\invariant_p) \in \textit{Systems} \})$
        \State $\hat{\actions} \gets \emptyset$
        \For{$j \in [1,\nrgua]$} \Comment{Generalize actions for guarantee $j$}
            \State $\textit{Strat}_j \gets \{ (p,\actions^j_p) \mid (p, M_{p}, \_) \in \textit{Systems}, \text{$A^j_p$ are the actions of $M_p$ for $\justiceguarantee_j$}  \}$
            \State $\hat{\actions}^j \gets \emptyset$ \Comment{Generalized substrategy $j$}
            \For{$(\_, u) \in \actions$}
                \State $\textit{g} \gets \textsc{GeneralizeExpr}(\{(p,g) \mid (p,\actions_p) \in \textit{Strat}_j \land (g,u) \in \actions_p \})$
                \State $\hat{\actions}^j \gets \hat{\actions}^j \cup \{(g,u)\}$
            \EndFor
            \State $\hat{\actions} \gets \hat{\actions} \cup \{ \hat{\actions}^j \}$
        \EndFor
        \State $(l,r) \gets \textsc{SynthesizeRankingFunction}(\invariant, \{\textit{Strat}_j\}_{j \in [1,n]})$
        \State $M \gets (\statevars \cup \{c\}, \inputvars, \init \land c=1, \hat{\actions})$
        \State $p_{\newp} = \textsc{CheckConsistency}(\game, M, \invariant, (l,r))$
        \If{consistent} \textbf{return} $M$  \EndIf    
        \State $P' \gets \{p_{\newp}\}$
    \EndWhile
\EndFunction

\end{algorithmic}
\end{algorithm}

\begin{example}\label{ex:generalization}
The  program from \cref{ex:runningintro} is constructed by this algorithm as follows.
We  focus only on the substrategy for $\cntr = 1$ (where $x$ is decreased).
As the \textbf{first parameter instantiation},  take $min=0 \land max = 2$.
    As our first substrategy,  assume we obtain $$M_{(0,2)}: (x+d\geq1, x'=x-1+d), (x+d<1, x'=x+d),$$ 
    along with invariant $\invariant_{(0,2)}: 0 \leq x \leq 2$.
    Since we have only one parameter value, the generalized strategy and the generalized invariant are just $M_{(0,2)}$ and $\invariant_{(0,2)}$. 
    When we verify whether $M_{(0,2)}$ works for all values of $\min$ and $\max$, we get the counterexample $\min=1 \land \max = 3$: $M_{(0,2)}$ reduces $x$ below $1$, which violates the proposed invariant.
    
    We continue with the \textbf{second iteration} for $(\min,\max) = (1,3)$. We find the substrategy $M_{(1,3)}: (x+d\geq2, x'=x-1+d), (x+d<2, x'=x+d)$
    with invariant $\invariant_{(1,3)}: 1 \leq x \leq 3$.
    We  generalize $M_{(0,2)}$ and $M_{(1,3)}$ to obtain  $$M_{\{(0,2),(1,3)\}}: (x+d\geq \mathbf{\text{\bf max}-1}, x'=x-1+d), (x+d< \mathbf{\text{\bf max} - 1}, x'=x+d).$$
    It is not hard to see that this program is not fully general. For instance for $(\min,\max) = (1,5)$, the system will never reach $x = 1$. We will thus continue the example in \cref{sec:rankingfunctions}, where we introduce ranking functions.

\end{example}
We will now state our main correctness result. The necessary lemmas will follow.
\begin{theorem}\label{thm:correctness}
Algorithm \textsc{ParameterizedSynthesis} is correct. That is, if it returns \textit{Unrealizable}, then $\game$ is unrealizable, and if it returns a system $M$, then $M \models \game$.
\end{theorem}
\begin{proof}
Correctness of the  \textit{Unrealizable} result follows from the fact that $\game$ is unrealizable iff for some $p$, we have that $\game_{p}$ is unrealizable. The correctness of a positive answer follows from the correctness of \textsc{CheckConsistency} (\cref{prop:checkconsistency}).
\end{proof}

Note that the algorithm may not return a result. However, for certain types of specifications, it does guarantee progress.
\begin{proposition}\label{prop:progress}
If 
GR(1) synthesis is decidable for $\game_p$ for any $p \in P$,
our synthesis algorithm ensures progress by constructing in every iteration, a larger set of parameters $P'$ and a system $M$ that is correct for $P'$.
\end{proposition}
\begin{proof} Sketch. A finite set of expressions, one for each $p \in P'$, can trivially be generalized using a case distinction over $p$. We can thus generalize the guards from $M_p$ for $p \in P'$ to obtain guards for a system $M_{P'}$, and that system will be correct for $P'$. We can do the same for invariants. As we will see in \cref{prop:existenceofrankingfunction}, a ranking function exists for any concrete system, and it can be generalized the same way. Finally, checking consistency is decidable if the underlying theory is (\cref{prop:checkconsistency}). 
\end{proof}
This result does not extend to $P$: a problem may be realizable for every single $p \in P$, but this does not imply that a program exists that works for all parameters.
For example, a guard $x \cdot p_1 \leq p_2$ is not in linear integer arithmetic, but all its concrete instances are.

\section{Expression Generalization}
\label{sec:generalization}

In this section, we describe a method to generalize a set of expression for different parameter values. As its input, the method takes a set $\exprs: P \rightharpoonup \exprs(\variables \setminus \paramvars)$ of expressions for the set $\dom(\exprs)$ of parameter values. It returns a single expression  $\expr \in \exprs(\variables)$ that includes parameter variables such that for each $p \in \dom(\exprs)$ we have $\expr[p] = \exprs(p)$.

The algorithm proceeds in two stages: first, we cluster similar expressions. For each cluster, we use anti-unification to identify the places in which the expressions differ, and use syntax-guided synthesis to synthesize a generalized expression for each hole. Second, when  no  similar expressions remain, %
we generalize the remaining dissimilar expressions by using case distinction over the parameter values, using interpolation to find proper parameter distinctions. 

The main data structure in our algorithm is a forest $\F = (N,\E, \exprmap, p)$, where 
$N$ is a set of nodes,  
$\E$ is a set of edges consisting of tuples $(n,s,n')$, where $s \in \mathcal{S}(\Holes)$ is a substitution,
$p: N \rightarrow 2^{\dom(\exprs)}$ associates a set of parameter values to each node, and
$\exprmap(n)$ is an expression in $\exprs(\variables \cup \Holes)$. Intuitively, the initial expressions form the leaves of our forest. Nodes higher up in a tree may contain holes that are the result of anti-unification and the substitutions on the edges explain how to reconstruct the leaves. Non-leaf nodes may also contain parameters and the instantiation of such an expression with the parameter valuation of a leaf results in an expression that is equivalent to the leaf.

We require that $E$ is forest-shaped. We denote the roots of a forest $\F$ by $\roots(\F)$, the leaves by $\leaves(\F)$, the subtree rooted in a node $n$ by $\tree(n)$, and its children by $\children(n)$. We will also write $\leaves(n)$ to denote $\leaves(\tree(n))$ and we  extend $\exprmap$ such that for $N' \subseteq N$, $\exprmap(N') = \{\exprmap(n) \mid n\in N'\}$.

We further require that for any leaf $p$, we have $|p(n)| = 1$ and otherwise, $p(n) = \bigcup_{n' \in \children(n)} p(n')$. Furthermore, for any edge $(n,s,n')$, we have that  $\exprmap(n') = \exprmap(n)[s]$. Thus, every node $n$ represents an expression $\exprmap(n)$, possibly with holes, that is valid for $p(n)$ and if $s_1,\dots s_n$ are the substitutions in the path from $n$ to $n''$, we have $\exprmap(n'') = \exprmap(n)[s_n \circ \dots \circ s_1]$. We write $E^* = \{(n_1,s,n_k) \mid \exists n_2,\dots,n_{k-1}. \forall i. (n_i, s_i, n_{i+1}) \in E \wedge s = s_{k-1} \circ \dots \circ s_{1}  \}$. In particular, for any $n \in N$, $(n,\id,n) \in E^*$. 

\begin{algorithm}[tbp]
\caption{Finding a generalization of a set of expressions}\label{alg:expr_clustering1}
\begin{algorithmic}
 \Function{GeneralizeExpr}{$\F = (N, \E, \exprmap, p)$} \Comment{Input: forest consisting of trees of one leaf each and no edges.}
     \State \clustersimilar($\F$)
     \State \generalizesimilar($\F$)
     \State \generalizedissimilar$(\F)$ 
     \State \textbf{return} $\lambda(\roots(\F))$ \Comment{singleton root}
 \EndFunction
\\
\Function{ClusterSimilar}{$\F = (N, \E, \exprmap, p)$} 
    \While{$\exists n_1, n_2 \in \roots(\F): n_1 \neq n_2 \text{ and } d(\exprmap(\leaves(n_1)) \cup \exprmap(\leaves(n_2))) \leq (0,M)$}
        \State $n_1,n_2 \gets \argmin_{n_1\neq n_2\in \roots(\F)} d(\exprmap(\leaves(n_1)) \cup \exprmap(\leaves(n_2)))$
        \State \merge$(n1,n2)$
    \EndWhile
\EndFunction
\\
\Function{Merge}{$\F= (N, \E, \exprmap, p), n_1, n_2$}
    \State $g,s_1,s_2 \gets \textsc{Anti-unify}(\lambda(n_1),\lambda(n_2))$
    \State $N \gets N \cup \{n\}$ \Comment{Where $n$ is a fresh node}
    \State $\lambda(n), p(n) \gets g, p(n_1) \cup p(n_2)$
    \State $\E \gets \E \cup \{(n,s_1,n_1), (n, s_2, n_2)\}$
\EndFunction
\\
\Function{GeneralizeSimilar}{$\F = (N, \E, \exprmap, p)$} 
    \For{$n \in \roots(\F)$}
        \For{each hole $h \in \exprmap(n)$}
            \State $\expr_{\synt} \gets \textsc{FuncSynth}(\expr_{\synt}, \bigwedge_{n' \in \leaves(n), (n,s,n') \in \E^*} \forall \state,\inputval: \expr_{\synt}(p(n'),\state,\inputval) = s(h))$ 
            \State $\lambda(n) \gets \exprmap(n)[h \gets \expr_{synth}]$
        \EndFor
    \EndFor
\EndFunction
\\

\Function{GeneralizeDissimilar}{$\F = (N, \E, \exprmap, p)$} 
    \State $R \gets \roots(\F)$
    \Function{Nodes${}_R$}{n}
        \State \Return{$R\cap \tree(n)$}
    \EndFunction
    \While{$|\roots(\F)| > 1$}
        \State $n_1,n_2 \gets \argmin_{n_1\neq n_2\in \roots(\F)} d(\lambda(\nodes_R(n_1)) \cup \lambda(\nodes_R(n_2)))$
        \State $N \gets N \cup \{n\}$ \Comment{Where $n$ is a fresh node}
        \State $\E \gets \E \cup \{(n,\id,n_1), (n,\id,n_2)\}$    
        \State $p(n) \gets p(n_1) \cup p(n_2)$
        \State $\text{cond} = \textsc{Interpolate}(p(n_1), p(n_2))$
        \State $\lambda(n) = \ite(\text{cond}, \exprmap(n_1), \exprmap(n_2))$ 
    \EndWhile
\EndFunction
\end{algorithmic}
\end{algorithm}

\paragraph{Similarity} %
The \emph{similarity} of a set of expressions is defined as a lexicographic order of two measures. 
Our primary measure is the \emph{variable distance}, viz. the number of variables from $\variables$ that are in some but not in all expressions: $d_v(\exprs) = |\bigcup_{\expr \in \exprs}\vars(\expr) \setminus \bigcap_{\expr \in \exprs}\vars(\expr)|$.
Our secondary measure is the size of the substitutions. Let $g^{\exprs}$ be the (unique) anti-unification of the set $\exprs$ and let $S^\exprs_{\expr} \subseteq \mathcal{S}(\Holes)$ denote the corresponding substitutions, where $g^\exprs[S^\exprs_{\expr}] = \expr$.
We define $d_s(\exprs) = \sum_{(h,\expr') \in S^\exprs_{\expr}}\max_{\expr \in \exprs}|\expr'|^2$ to be the sum of the squares of the biggest substituted expression for each hole. We square the size of the substituted expressions because we prefer multiple smaller holes over few big ones.
Our combined measure is the lexicographic order $d(\exprs) = (d_v(\exprs), d_s(\exprs))$. We call a set of $\exprs$ of expressions \emph{similar} if $d(\exprs) \leq (0,M)$, where $M$ is a user supplied constant. Thus, similar expressions are expressions that contain the same variables and  can be unified  with relatively small substitutions.

\paragraph{Generalizing Similar Nodes}
Our algorithm is described in \cref{alg:expr_clustering1}. The input is a set of ground expressions with the associated parameter values, in the form of a forest consisting of isolated nodes. In the first stage, we focus on similar expressions.  We build a tree by iteratively clustering subtrees that represent sets of similar expressions. %

After clustering has completed, we generalize similar expressions using functional synthesis. First, we use anti-unification to find a generalizer, a set of holes, and the corresponding substitutions. We then construct a new parent node for the similar expressions that is labeled with the generalizer and label the edges with the substitutions.
Let $\Holes_n$ be the set of holes in $\exprmap(n)$. For each $h \in \Holes_n$, we use syntax-guided synthesis to synthesize an  expression $s^*(h)$ over $\paramvars$ that generalizes all substitutions, thus constructing a substitution $s^*: \Holes_n \rightarrow \expr(\variables)$. The ground expression $\exprmap(n)[s^*]$ generalizes all children: for each $(n,s,n') \in \E$, we have  
$\exprmap(n)[s^*][p(n')] = e(n')$. 

The syntax-guided synthesis query uses the following grammar to synthesize expressions $\expr_{\synt}$ for holes in \cref{alg:expr_clustering1}. The root non-terminal is $H$ \footnote{For real arithmetic this grammar might not always include a generalizing expression. To guarantee termination an appropriate grammar needs to be used that can at least encode a case distinction over all concrete values. If no grammar is provided the syntax-guided synthesis tool it will select a grammar based on the used theory. For progress and termination arguments we will assume that the used grammar always contains a generalizing expression.}.
\begin{align*}
      H &::= I \mid B,
    \\B &::= B \land B \mid \neg B \mid  I \leq I \mid \statevars_{\mathbb{B}} \mid \inputvars_{\mathbb{B}},
    \\I &::= I + I \mid I - I \mid ite(B,I,I) \mid V \mid  0 \mid 1, \text{ and}
    \\V &::= \paramvars \mid \statevars_{\mathbb{Z}} \mid \inputvars_{\mathbb{Z}} \mid \statevars_{\mathbb{R}} \mid \inputvars_{\mathbb{R}}.
\end{align*}

\paragraph{Generalizing Dissimilar Nodes}
In the second stage, %
we cluster the remaining nodes. We greedily create a new parent for the pair of most similar nodes and run anti-unification to factor out the common structure of the expressions. We do this until we are left with a single binary tree. Since the corresponding expressions are not syntactically alike, we do not  generalize the substitutions for the holes to a simple unifying expression. Instead, a we construct a case distinction over the parameter values. Since not every parameter value occurs in the tree, we use interpolation to  judiciously guess the values for the case distinction. 

We use CVC5 to compute interpolants in \cref{alg:expr_clustering1} and restrict them using the following grammar with the root non-terminal $X$.
\begin{align*}
      X &::= A \mid A \lor X,
    \\A &::= C \mid C \land A,
    \\C &::= I = I \mid I \leq I, \text{ and}
    \\I &::= \paramvars \mid \textit{constants used in parameter valuations}.
\end{align*}
The grammar ensures that the interpolation engine attempts to build a short and general expression.
An interpolant can always be found for a given parameter sets as the interpolant may consist of a simple enumeration of all values from one of the parameter sets.

\begin{example}
\Cref{fig:generalization_tree} shows the generalization tree constructed for four expressions. Each node $n$  shows the expression $\lambda(n)$ in blue, the associated parameters $p(n)$  in green and the generalized expression in red.
The first merge operation combines the expressions for $p=2$ and $p=1$ (bottom left), which only differ in their constants.
Their parent node shows the expression obtained by anti-unification with the holes $h1$ and $h2$. (We omit the substitutions.)
In the next step, we combine the nodes for $p=-2$ and $p=-3$.
Using syntax-guided synthesis, we find the  expressions $a \geq p \wedge a \leq p+2 $ on the left, a simple expression that generalizes the expression of its children: it equals $a \geq 2 \wedge a \leq 4$ when $p=2$ and $a \geq 1 \wedge a \leq 3$ when $p=1$. Similarly, we find the generalized expression $a \leq p \vee a \geq p+2$ on the right.

The resulting nodes have dissimilar expressions. The algorithm combines them into a single root node with  expression $h5$. Using interpolation, it  finds the distinguisher $p \geq 1$, which differentiates $p \in \{1,2\}$ and $p \in \{-2,-3\}$. 
The resulting expression is $\ite(p \geq 1,\ a \geq p \wedge a \leq p + 2,\ a \leq p \vee a \geq p+2)$.
\end{example}

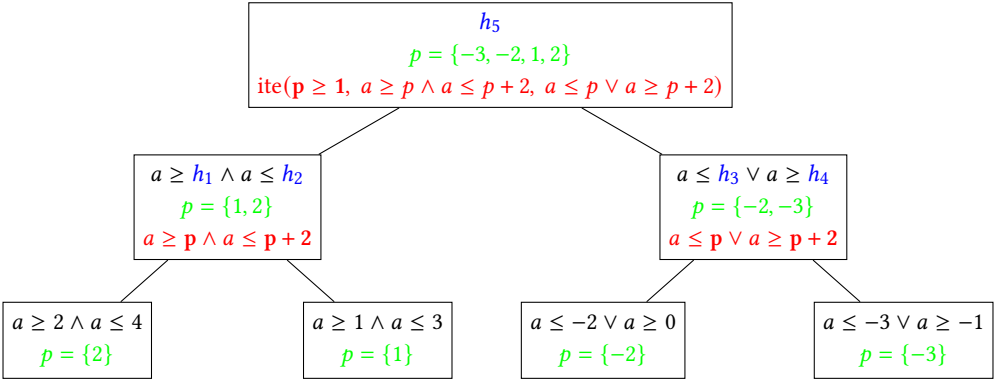
\begin{figure}[tb]
    \centering
\begin{tikzpicture}[level distance=2cm, sibling distance=8cm, level 2/.style={sibling distance=4cm}, every node/.style={draw, rectangle, minimum size=1cm, align=center}]
    \node (root) at (0,0) {
    \textcolor{blue}{\small $h_5$}\\
    \textcolor{green}{\small $p = \{-3,-2,1,2\}$}\\
    \textcolor{red}{\small $\ite(\mathbf{p \geq 1},\ a \geq p \wedge a \leq p + 2,\ a \leq p \vee a \geq p+2)$}
    };
    
    \node (node0) at (-3.5,-2) {\small  $a \geq \textcolor{blue}{h_1} \wedge a \leq \textcolor{blue}{h_2}$\\ 
    \textcolor{green}{\small $p = \{1,2\}$}\\
    \textcolor{red}{\small $a \geq \mathbf{p} \wedge a \leq \mathbf{p + 2}$}
    };
    \draw (root) -- (node0);
    
    \node (node00) at (-5.5,-3.75) {\small  $a \geq 2 \wedge a \leq 4$ \\
    \textcolor{green}{\small $p = \{2\}$}
    };
    \draw (node0) -- (node00);
    
    \node (node01) at (-1.5,-3.75) {\small  $a \geq 1 \wedge a \leq 3$ \\  
    \textcolor{green}{\small $p = \{1\}$}};
    \draw (node0) -- (node01);
    
    \node (node1) at (3.5,-2) {\small  $a \leq \textcolor{blue}{h_3} \vee a \geq \textcolor{blue}{h_4}$ \\ 
    \textcolor{green}{\small $p = \{-2,-3\}$}\\
    \textcolor{red}{\small $a \leq \mathbf{p} \vee a \geq \mathbf{p+2}$}
};
    \draw (root) -- (node1);
    
    \node (node10) at (1.5,-3.75) {\small  $a \leq -2 \vee a \geq 0$  \\ 
    \textcolor{green}{\small $p = \{-2\}$}};
    \draw (node1) -- (node10);
    
    \node (node11) at (5.5,-3.75) {\small  $a \leq -3 \vee a \geq -1$\\ 
    \textcolor{green}{\small $p = \{-3\}$}};
    \draw (node1) -- (node11);   
\end{tikzpicture}
   
    \caption{Example of a Generalization Tree}
    \label{fig:generalization_tree}
\end{figure}

\begin{proposition}
\label{prop:generalization}
Given a set of expressions $\exprs: P \rightharpoonup \exprs(\variables \setminus \paramvars)$ with $P' = \dom(\exprs)$,  a forest $\F = (N, \emptyset, \exprmap, p)$ consisting of exactly one leaf $n_p$ for each $p \in P'$  with $p(n_p) = p$,
\textsc{GeneralizeExpr}($\F$) returns an expression $\lambda \in \exprs(\variables)$ such that for all $p \in \params'$, we have $\lambda[p] =  \exprmap(n_p)$.
\end{proposition}
\begin{proof}
    The algorithm maintains the following invariants, which we will not prove separately: (1) we have a forest, (2) the set of leaves remains unchanged and for each leaf $n$,  $\lambda(n)$ and $p(n)$ remain unchanged, and (3) for any $n$, $p(n) = \bigcup_{n' \in \leaves(n)}p(n')$ and the sets $p(n')$ for the children $n'$ of some node $n$ are disjoint.
    
    To be a bit more precise, in this proof, we will use $=$ for syntactic equivalence and $\equiv$ for semantic equivalence.

    Function \clustersimilar maintains the invariant $I_1$: for any node $n$ and any leaf $n' \in \leaves(n)$, if $(n,s,n') \in E^*$, we have that $\lambda(n)[s] = \lambda(n')$. Note that this invariant holds trivially for the input because $(n,\id,n) \in E^*$ for all $n$. Further,  by induction, if  we have $\lambda(n)[s] = \lambda(n')$ for all $(n,s,n') \in E$, then $\lambda(n)[s] = \lambda(n')$ for all $(n,s,n') \in E^*$. It thus suffices to prove the fact for all edges.

    \merge and, by extension, \clustersimilar maintain $I_1$ by definition of anti-unification: for the new edge $(n,s_1,n_1)$, we have that $\lambda(n)[s_1] = g[s_1] = \lambda(n_1)$ and similarly for $(n,s_2,n_2)$. After \clustersimilar completes, for any root $r$, we have that $\lambda(r)$ is an expression over $\variables \setminus \paramvars \cup \Holes$ and $I_1$ holds.

    Function \generalizesimilar does not change the structure of the forest, but removes the holes from the labelings and replaces them with expressions over $\variables$, now including $\paramvars$. As an invariant of the loop, we have $I_2(n)$, which says that for all $n' \in \leaves(n)$, we have $\lambda(n)[s][p(n')] \equiv \lambda(n')$. Initially, this follows from $I_1$ and the fact that $\lambda(n)$ does not contain any parameters. It holds inductively during the loop, because at every iteration, one hole $h$ is replaced by an expression $\expr_{\synt}$ such that 
    \[
       \bigwedge_{n' \in \leaves(n), (n,s,n') \in \E^*} \forall \state,\inputval.\, \expr_{\synt}(p(n'),\state,\inputval) \equiv s(h).  
    \]
    At the end of \generalizesimilar, $\lambda(r)$ is free of holes and thus, $I_2$ implies $I_3$, which we define as $\forall r \in \roots(\F).\, \forall n' \in \leaves(n).\, \lambda(n)[p(n')] \equiv \lambda(n')$. 

    It is not hard to see that \generalizedissimilar maintains $I_3$. At every iteration of the loop, it introduces a new node $n$ with children $n_1$ and $n_2$ so that $\lambda(n)$ is equivalent to $\lambda(n_i)$ for any $p \in n(p_i)$.  Invariant $I_3$ follows by transitivity of the equivalence relation.   Finally, the algorithm returns $\lambda(r)$ for the single root $r$ and the proposition follow from $I_3$.
\end{proof}

If the expressions to be generalized contain multiplications with integer constants, the generalization algorithm may return a generalized expression that contains nonlinear integer arithmetic. 
For example, generalizing $2\cdot x$ for $p=2$ and $3 \cdot x$ for $p=3$ would result in the general expression $p\cdot x$, which is not linear.
To ensure the expression remain in a decidable logic we use a preprocessing step that rewrites any multiplication with integer constants into repeated addition. To generalize additions of different length, our algorithm uses the generalization of dissimilar expression and will essentially build a big case distinction. 

While this approach guarantees that all expressions remain in a decidable logic, it also means that a generalizer like $p \cdot x$ is never found. In  practice allowing non linear expressions can lead to better results, but it has to rely on an SMT solver's support for an undecidable logic which is not guaranteed to work. In our decidability proofs, we assume that generalization does not produce expressions in an undecidable fragment.

\paragraph{Optimizations}
To improve the performance of the generalization approach we use some additional optimizations.
Before using the generalization algorithm discussed above an SMT solver is used to check if the expression are equivalent in relevant context e.g. for all states inside the winning region.
If they are semantically equivalent (even if not syntactically) the first expression is used as the generalizer.

The generalization results are also cached. 
When a new parameter instance is added we verify if the result for the previous instances also generalizes the new instance in that case it is reused.
This is especially helpful as a counter example often invalidates only some parts of the strategy and the cache avoids computing generalizations which do not need to change.

As a preprocessing step expressions are normalized by rewriting inequalities into a normal form and sorting the parameters of commutative functions.
Consider the following expressions from linear integer arithmetic:
\begin{align*}
p=2&,\hspace{6pt} x \geq 2 \land 2 \geq y\\
p=3&,\hspace{6pt} y < 4 \land x \geq 3.
\end{align*}
All inequalities are rewritten to non-strict inequalities with the variables on the left hand side.
The arguments of the conjunction are sorted by lexicographically comparing the expressions:
\begin{align*}
p=2&,\hspace{6pt} x \geq 2 \land y \leq 2\\
p=3&,\hspace{6pt} x \geq 3 \land y \leq 3.
\end{align*}
These expressions are nicely aligned and can be generalized to $x \geq p \land y \leq p$.

\section{Verification}
\label{sec:proofs}

Besides a parameterized program, we need a proof of correctness.
The  programs produced by our approach are challenging. In initial experiments with existing verification tools \cite{DBLP:conf/cav/DietschHLP15,DBLP:journals/pacmpl/UnnoTGK23}, we could not verify, or  produce counterexamples even for simple systems.

We present a heuristic verification approach that is specialized to our setting.
Inspired by the proof witnesses used in software model checking \cite{Beyer22}, our approach distills the information from the concrete game solving algorithms into a symbolic proof candidate that we  check for consistency with the parameterized program.
If the proof fails, a new parameter value is produced as a counterexample.
In \cref{sec:decidability}, we prove that the verification problem for parameterized programs is undecidable. %
As we will see, deciding whether a witness is correct or not is decidable, but we cannot guarantee that we find the right witness.

\subsection{Invariants}
A strategy computed using GR(1) game solving consists of sub-strategies for each of the guarantees.
We use inductive invariants to prove that a strategy is safe.

\begin{definition} \label{def:invariant} 
Given a 
parameterized game $\game = (\paramvars,\statevars,\inputvars, \init, \envtrans, \actions, \wincond)$, such that $\systrans$ is the transition relation of $\actions$,
and a parameterized system $M = (\paramvars,\statevars\cup\{c\},\inputvars,\init\land c=1, \bigcup_{j=1}^{n} \hat{A}^j)$, such that $\hat{\systrans} = \bigvee_{j=1}^{n} \hat{\systrans}^j$,
a predicate $\invariant \in \exprs(\statevars)$ is an \emph{invariant} for $M$ and $\game$ iff it satisfies $\invpred$, which states that for each $j$, the conjunction of the following requirements holds.
\begin{align}
& init \rightarrow \invariant,\label{eq:inv_init}
\\ & (c = j \land \invariant \land \neg\justiceguarantee_j \land \envtrans \land \hat{\systrans}^j) \rightarrow (c' = j \land \invariant'),\label{eq:inv_preservation}
\\ & (c = j \land \invariant \land \justiceguarantee_j \land \envtrans \land \hat{\systrans}^j) \rightarrow (c' = j\oplus 1 \land \invariant'),\label{eq:inv_preservation_switch}
\\ & (c = j \land \invariant \land \envtrans \land \hat{\systrans}^j) \rightarrow \systrans\text{, and}\label{eq:inv_safe}
\\ & \exists \state',c': \invariant \land \envtrans \land \hat{\systrans}^j.\label{eq:inv_deadlock_free}
\end{align}
\end{definition}
\Cref{eq:inv_init} ensures that every initial state is safe, \cref{eq:inv_preservation,eq:inv_preservation_switch} ensure that a substrategy preserves the invariant when it remains in the same substrategy and when it switches to the next one, resp.
\Cref{eq:inv_safe} ensures that all actions performed by the strategy are allowed by $\game$, and
\Cref{eq:inv_deadlock_free} ensures that the strategy does not deadlock.

We  obtain the invariant $\invariant_p$ for $M_p$ for each parameter $p \in P'$ as the winning region from the game solving algorithm. We use the generalization algorithm from \cref{sec:generalization} to generalize the $\invariant_p$'s into an invariant candidate $\invariant$.

\subsection{Ranking Functions}\label{sec:rankingfunctions}
We use ranking functions to prove that the justice guarantees are satisfied infinitely often if the  assumptions are satisfied infinitely often.
The ranking functions are defined on the level of sub-strategies, where each substrategy has to ensure it eventually satisfies its associated guarantee, unless one of the assumptions is only satisfied a finite number of times. The invariant generation routine generalizes the invariants of each concrete system into a parameterized one. A first idea may be to similarly generalize the iterates $mX[j][r]$ for each $j$ and $r$. This idea, however, does not work as the number of iterates may be different for each parameter instantiation. Instead, we use syntax-guided synthesis to construct a ranking function that maps each state to a pair of integers. 

We define  a parameterized ranking function for each justice guarantee $\justiceguarantee_j$, where $r_j(p,x)$ contains the number of steps to fulfilling $\justiceguarantee_j$, and $l_j(p,x) = i$ iff $r_j$ is guaranteed to decrease when $\justiceassumption_i$ is satisfied.
\begin{definition}\label{def:rankingfunction}
Given a set of parameters $P'$, a parameterized game $\game$
 with invariant $\invariant$, and a parameterized system $M$ %
 with transition relation $\hat{\systrans} = \bigvee_{j=1}^{n} \hat{\systrans}^j$,  families $l=\{l_j\}_{j \in [1,n]} $ and $r = \{r_j\}_{j \in [1,n]}$ with $r_j : \params \times \states \to \mathbb{N}$ and  $l_j: \params \times \states \to [1,\nrass]$  form a \emph{$P'$-ranking function}
if they satisfy the predicate $\rankpred$, which states that for all $p \in P'$ and $j \in [1,n]$, we have
\begin{align}
&\invariant \rightarrow r_j(\paramval,\state) \geq 0, \label{eq:rankpos}\\
&\textstyle(\invariant \land \neg \justiceguarantee_j \land \envtrans \land  \hat{\systrans}^j \land \bigvee_{i\in[1,\nrass]} (l_j(\paramval,\state)=i \land \justiceassumption_i')) \rightarrow r_j(\paramval,\state) > r_j(\paramval,\state'),\label{eq:rankdecrease}\\
&(\invariant \land \neg \justiceguarantee_j \land \envtrans \land  \hat{\systrans}^j) \rightarrow r_j(\paramval,\state) \geq r_j(\paramval,\state')\text{, and}\label{eq:ranknondecrease}\\
&(\invariant \land \neg \justiceguarantee_j \land \envtrans \land \hat{\systrans}^j \land r_j(\paramval,\state) = r_j(\paramval,\state')) \rightarrow l_j(\paramval,\state) \geq l_j(\paramval,\state').\label{eq:ranknextassump}
\end{align}
A \emph{$P'$-ranking} function is a \emph{ranking function} if $P' = P$.
\end{definition}

We use a syntax-guided synthesis solver to search for small candidate ranking functions $(r_j,l_j)$ that are valid for a set of concrete sub-strategies.
The function\\ $\textsc{SynthesizeRankingFunction}(\textit{sub-strategies})$ used in \cref{alg:param_synthesis} can be implemented as\\ $\textsc{FuncSynth}(r_{\synt}, l_{\synt}, \bigwedge_{(\paramval, \hat{\systrans}^j) \in \textit{sub-strategies}} \forall \state,\inputval,\state': \rankpred)$.
This quantifier-free synthesis query uses the same background theory as the specification and is thus decidable iff the background theory is.
The following proposition shows that  a ranking function exists for any  concrete system. 
\begin{proposition}\label{prop:existenceofrankingfunction}
For any parameterized GR(1) specification $\game$ and any concrete system $M_p$ that is synthesized from $\game_p$,  a ranking function exists.
\end{proposition}
\begin{proof}
    Let $\sigma$ be the transition relation of $M_p$ and $\invariant$ be an invariant. Consider a  substrategy $\hat{\systrans}^j_p$ in $M_p$.
    The invariant $\invariant$ and the ranking sequence $\{\invariant^j_k\}_{k \in [1, L_j]}$ computed by the \textsc{GR(1)-Synthesis} algorithm (\cref{sec:gr1}) for $M_p$ can be used to define $r_p^j(x) = k \leftrightarrow x \models \invariant^j_k \setminus \invariant^j_{k-1}$.
    Let %
    $\{X^i_k\}_{i=1}^{\nrass}$ be a partitioning of $\invariant^j_k \setminus \invariant^j_{k-1}$ such that for every state in $X^i_k$ the assumption $\justiceassumption_i$ guarantees progress.
    Construct $l^j_p(x) = i \leftrightarrow x \in \bigcup_{k\in[1, L_j]} X^i_k$.
    By the construction of the game solving algorithm, $\invariant$ and ranking function $(r_p^j, l^j_p)$ satisfy $\invpred$ and $\rankpred$, resp., for parameter set $\{p\}$.
\end{proof}

The previous proposition proves that the standard way to synthesize a GR(1) system immediately gives us a ranking function. 
However, this ranking function has some problems when it comes to generalization: it only works if the number of steps required is less than or equal than in the concrete systems used to construct it, the decrease in each step cannot be larger than $1$, and it relies on a large case distinction which does not generalize well.
\Cref{def:rankingfunction} allows for a larger set of functions: we may construct a correct ranking function that is not a generalization of the ranking functions of the individual systems in the sense of \cref{prop:existenceofrankingfunction}. We have found that in practice, it is easier to find simple ranking functions in this larger set.

Note also that there may not be any  ranking function for a given correct parameterized system. However, \emph{if} we find an invariant and  ranking function for a parameterized system, we have proven its correctness. 

We use the following grammar to synthesize ranking functions separately for each substrategy.
The root non-terminals are $R$ for $r_j$ and $L$ for $l_j$. 
\begin{align*}
      R &::= S \mid S + R
    \\S &::= V \mid abs(T - T)
    \\T &::= X \mid X + T
    \\X &::= V \mid 1
    \\L &::= J \mid ite(B, L , L)
    \\B &::= I \leq I \mid \statevars_{\mathbb{B}}
    \\I &::= I + I \mid I - I \mid V \mid  0 \mid 1 
    \\J &::= 1 \mid \dots \mid \nrass
    \\V &::= \paramvars \mid \statevars_{\mathbb{Z}} \mid \lfloor \statevars_{\mathbb{R}} \rfloor
\end{align*}
The grammar for $R$ resembles linear ranking functions\cite{DBLP:conf/pods/SohnG91, DBLP:conf/vmcai/PodelskiR04, DBLP:conf/tacas/ColonS01}, but restricted to positive integer coefficients. It is also extended with an absolute difference operation, because reducing the distance to a target is common pattern in reactive systems. The grammar for $L$ encodes decision trees to select one of the available assumptions.
For specifications with zero or one assumptions we only synthesize $r$ and hard code $l$ to point at the single assumption or at an artificial assumption $True$.
The implementation sacrifices the theoretical progress guarantee for faster results and more scalability in practical cases.
For progress and termination arguments we assume a grammar that can also encode the construction in \cref{prop:existenceofrankingfunction} via case distinction.

\subsection{Consistency}
We define $\textsc{CheckConsistency}(\game, M, \invariant, (l,r))$ to return true if 
\[\game, M,  \models \invpred \wedge \rankpred,\]
and a counterexample otherwise. A counterexample is a parameter instantiation $p \in P$ such that $\game_p, M_p,  \not\models \invpred$ or $\game_p, M_p,  \not\models \rankpred$. Since the invariant and the ranking function are each correct for our current set of parameters $P'$, for any counterexample $p$, we know that  $p \notin P'$. We perform this check by verifying $\invpred$ and $\rankpred$ using an SMT solver. Both are the quantifier-free predicates  over the underlying theory with free variables in $\variables$ (including $\paramvars$). The following proposition shows that the formula is decidable and indeed checks whether the system is correct. 

\begin{proposition}\label{prop:checkconsistency}
Suppose the grammar for the invariants and ranking functions are included in a decidable logic, then checking consistency is decidable. 
If $\textsc{CheckConsistency}(\game, M,\invariant, (l,r))$ returns true, then $M \models \game$.
\end{proposition}
\begin{proof}
    Checking consistency is decidable by the fact that the grammar results in decidable formulas.
    
    We want to prove that if \textsc{CheckConsistency} returns true, then for any $p$, $M_p \models \game_p$, as defined in \cref{def:gr1semantics}. Let $\pi = x_0, \inputval_0, x_1, \inputval_1, \dots$ be a trace of $M_p$ and assume that for all $i$ we have that  $x_i, \inputval_{i+1},x_{i+1} \models \envtrans$. First, by \cref{eq:inv_init,eq:inv_preservation,eq:inv_preservation_switch}, $\invpred$ ensures that $\invariant$ is an inductive invariant. By \cref{eq:inv_deadlock_free}, $M_p$  is deadlock-free, and by \cref{eq:inv_safe}, the system transition constraint holds.

    It remains to prove that the GR(1) condition is fulfilled. The crucial point is that when the system applies substrategy $\hat{\systrans}^j$, there cannot be an infinite trace $\pi = x_0, \inputval_0, x_1, \inputval_1, \dots$ that always fulfills the environment transition relation, on which the all justice assumptions are fulfilled infinitely often and justice guarantee $j$ is not fulfilled at all. By contradiction, suppose that $\pi$ is such a trace. By \cref{eq:rankpos,eq:ranknondecrease}, $r_j$ is at least zero and monotonically decreasing . Hence, it must remain constant from some point on. By \cref{eq:ranknextassump}, this means that $l_j$ is also constant from some point on. Assume  that  $l_j(p,x_k) = i^*$ for all $k \geq k^*$. Now by \cref{eq:rankdecrease}, $r_j(p,x_k)$ decreases whenever $x_k \models \justiceassumption_{i*}$, which is infinitely often. This contradicts the fact that $r$ is constant.
\end{proof}
This concludes that a positive answer from $\textsc{CheckConsistency}$ implies that the generalized system is correct. On the other hand, if the check fails, then for some $p$ either $\invariant$ is  not an invariant or $(l,r)$ is not a ranking function. We extract this $p$ from the satisfying assignment that the SMT solver returns and use it to refine our result.

\begin{example}\label{ex:ranking}

\newcommand{\exa}{\{(0,2),(1,3)\}}
We ended \cref{ex:generalization} with
$$M_{\exa}: (x+d\geq \max-1, x'=x-1+d), (x+d < \max-1, x'=x+d)$$
with invariant $\invariant_{\exa}: \min \leq x \leq \max$. 
For this system and  assumption $\infoften d < 0$, we compute the candidate ranking function $l(x) = 0$, $r(x) = x$. The generalized system does not fulfill the ranking function. For instance, for parameters $(\min,\max) = (-4,0)$, if $x=-3$ and $d =1$, the value of $r$ is negative, which contradicts \cref{eq:rankpos}.

\newcommand{\exb}{\{(0,2),(1,3),(-4,0)\}}
   
   Based on the counterexample, we build a \textbf{third substrategy} $M_{(-4,0)}: (x+d \geq -3, x' = x-1+d), (x+d < -3, x' = x+d)$ with invariant $\invariant_{(-4,0)}: -4 \leq x \leq 0$. We generalize the three concrete strategies to 
    $$M_{\exb}: (x+d\geq min+1, x'=x-1+d), (x+d<min+1, x'=x+d),
    $$
     invariant $\invariant_{\exb}: \min \leq x \leq \max$
    and ranking function: $l(x) = 0$, $r(x) = x-\min$.
    This corrects both the guard condition and the ranking function, which results in successful verification and a working system, which makes us happy.

Note that that the number of steps to reach $x = min$ depends on $max-min$. Thus, each concrete instances has a different number of iterates to reach the goal and for the general program, $r$ is not bounded. Thus, we cannot use a straightforward generalization of the iterates for each parameter value.
\end{example}

\section{Decidability}\label{sec:decidability}
\begin{proposition}
    The verification problem for parameterized infinite-state reactive synthesis are undecidable, even if all concrete problem instances are %
    decidable. 
\end{proposition}
\begin{proof}
    Sketch: We reduce the halting problem of 2-counter machines, which is undecidable, to parameterized synthesis and to the verification problem.
    We use a parameter $\max$ to limit number of instructions that is executed. %
    We use Minsky machines, 2-counter machines with the instruction set $\mathit{Inc}(r)$ (increment register $r$) and $\mathit{JZDec}(r,z)$ (if register $r$ is zero go to instruction $z$ otherwise decrement $r$ and go to the next instruction) \cite{Minsky61,DBLP:books/daglib/0008455}.
    
    From a 2-counter machine $C$, we construct a parameterized system $M^{C}$ as follows. System  $M^{C}$  has a parameter $\max$, state variables $r_1, r_2$ for the two registers, an instruction pointer $l$, and runtime counter $c$, which are all natural numbers. The instruction pointer may point to an line in $C$, equal  $\mathit{Halt}$ (the machine has terminated) or $\mathit{Loop}$ (the machine has not terminated within $\max$ steps).
    The system contains an action for each line $k$ of $C$ with guard $l = k$ that updates $r_1$, $r_2$, and $l$, and increments $c$.
    When $c = \max$, we set $l = \mathit{Loop}$; if the machine reaches the end of the program, $l$ is set to $\mathit{Halt}$. If $l \in \{\mathit{Loop}, \mathit{Halt}\}$, all variables retain their value.
    The  specification states that  $l \neq \mathit{Halt}$.

    The counter machine $C$ halts within $n$ steps iff the concrete system $M^{C}_n$ violates the specification. Thus, $C$ halts iff the parameterized system $M^{C}$ violates the specification, which concludes the reduction.

\end{proof}

\begin{proposition}
    The  realizability problems for parameterized infinite-state reactive synthesis are undecidable, even if all concrete problem instances are %
    decidable. 
\end{proposition}
\begin{proof}
    For the reduction to parameterized realizability, we use the same  specification and the same set of actions as in the previous proposition. There is effectively only one machine $M^C$ that fulfills the specification. 
    A synthesized system $M'$ uses actions with guards that imply a guard of an action that is included in the specification. Thus, an action in $M'$ can only be used for the line for which it is intended. 
    Further limiting a guard is not useful (except to exclude unreachable states), because $M'$ is not allowed to deadlock.     
    Thus, any synthesized system is equivalent to $M^C$, and the specification is realizable iff $C$ halts.
\end{proof}

\section{Evaluation}\label{sec:evaluation}
We have implemented our approach using Python in the prototype tool \toolname
\footnote{Artifact available online: \url{doi.org/10.5281/zenodo.15208031}}.
We use the gr1mt library \cite{DBLP:conf/isola/MaderbacherWB24} to solve concrete GR(1) games. Gr1mt uses the Z3 solver \cite{demoura2008} for quantifier elimination.
The syntax guided synthesis problems are handled by the CVC5 \cite{DBLP:conf/tacas/BarbosaBBKLMMMN22} solver.

Our implementation contains several optimizations, which we summarize here. 
If all guards in the specification share the same constraint, we use this constraint as a candidate for the general invariant.
Furthermore, if the guards for individual actions in the specification do not depend on the parameter, we use the first synthesized concrete system as a candidate for the generalized system. 
In all other cases, we initialize $P'$ with two valuations that are different for every parameter variable.

\subsection{Parameterized Synthesis Benchmarks}
We present a  set of new  parameterized specifications.
We compare our tool \toolname  with the infinite-state synthesis tools  Issy \cite{DBLP:journals/corr/abs-2502-03013} and Sweap \cite{azzopardi2025}, adding the parameters as state variables that are chosen by the environment at the first time step and remain constant.
We do not compare to gensys-ltl \cite{SamuelDK21,DBLP:conf/kbse/SamuelDK23} or Raboniel \cite{DBLP:conf/fmcad/MaderbacherB22}, as these tools cannot handle unbounded strategy loops, i.e., the situation in which the number of steps to reach a liveness constraint cannot be bounded a-priori, but depends on the value of unbounded variables.
Likewise, Temos \cite{ChoiFPS22} can only handle unbounded strategies that repeat a single statement and use only predicates from the specification. Also, both Issy and sweap outperform it.
We compare to Issy instead of rpgsolve \cite{DBLP:journals/pacmpl/HeimD24}, rpg-STeLA \cite{DBLP:conf/cav/SchmuckHDN24}, and tslmt2rpg \cite{DBLP:journals/pacmpl/HeimD25}, as these algorithms are all implemented in the Issy tool by the same authors. For Issy we test both realizability checking and synthesis and acceleration with uninterpreted functions and geometric acceleration. For each instance, we report the results for the faster acceleration configuration for both synthesis and realizability.

We summarize the benchmarks and experimental results in \cref{tab:param_synth_benchmarks}.
For each benchmark, we include the domain of the state and input variables as $\mathbb{Z}$ or $\mathbb{R}$, the number of parameters, state variables, and input variables, as well as the number of justice assumptions and guarantees.
The result columns lists the runtime required by each tool for synthesis or ``TO'' when the time out was exceeded.
The column $|P'|$ shows how many concrete systems where required by \toolname to construct the parameterized system.

\begin{table}
    \centering
    \caption{Parameterized synthesis benchmarks. Timeout: 5 minutes}
    \label{tab:param_synth_benchmarks}
    \small%
    \begin{tabular}{lllllll|rr|rrr}
    \toprule
        \textbf{Benchmark} & \textbf{T.} & $\mathbf{|\paramvars|}$ & $\mathbf{|\statevars|}$ & $\mathbf{|\inputvars|}$ & $\mathbf{\nrass}$ & $\mathbf{\nrgua}$ & $\mathbf{|P'|}$ & \textbf{Prasanva}  & \textbf{sweap} & \textbf{Issy (r.)} & \textbf{Issy (s.)} \\
        \midrule
        intro example                   & $\mathbb{Z}$ & 2 &  1 & 1 & 2 & 2 & 3  &    3.8s  &  TO     &  TO  & TO  \\
        intro example v2                & $\mathbb{Z}$ & 2 &  1 & 1 & 1 & 2 & 3  &    4.9s  &  TO     &  TO  & TO \\ %
        intro example r.                & $\mathbb{R}$ & 2 &  1 & 1 & 2 & 2 & 4  &    6.6s  &  ---      &  TO  & TO \\
        intro example r. v2             & $\mathbb{R}$ & 2 &  1 & 1 & 1 & 2 & 3  &    5.3s  &  ---      &  TO  & TO \\ %
        minimal example                 & $\mathbb{Z}$ & 1 &  1 & 1 & 1 & 2 & 3  &    1.5s  &  2.4s   &  TO  & TO \\
        distinct strategies             & $\mathbb{Z}$ & 1 &  1 & 1 & 1 & 2 & 3  &    1.7s  &  TO     &  TO  & TO \\
        win by assume violation         & $\mathbb{Z}$ & 2 &  1 & 0 & 1 & 2 & 3  &    3.2s  &  10.8s  &  TO  & TO \\
        single gate                     & $\mathbb{Z}$ & 2 &  1 & 1 & 1 & 2 & 3  &    4.3s  &  TO     &  TO  & TO \\
        double gates                    & $\mathbb{Z}$ & 2 &  1 & 2 & 2 & 2 & 4  &   14.6s  &  TO     &  TO  & TO \\
        two lanes gate                  & $\mathbb{Z}$ & 1 &  2 & 2 & 2 & 2 & 2  &    8.2s  &  TO     &  TO  & TO \\
        gate wind                       & $\mathbb{Z}$ & 1 &  1 & 2 & 2 & 2 & --  &      TO  &  TO     &  TO  & TO \\
        fetch1d                         & $\mathbb{Z}$ & 1 &  2 & 1 & 0 & 1 & 1  &    1.0s  &  TO     &  83s & TO \\ 
        fetch1d gate                    & $\mathbb{Z}$ & 1 &  2 & 2 & 1 & 1 & 2  &    5.2s  &  TO     &  TO  & TO \\
        fetch2d                         & $\mathbb{Z}$ & 1 &  4 & 2 & 0 & 1 & 1  &    2.0s  &  TO     &  TO  & TO \\
        fetch2d stop signal             & $\mathbb{Z}$ & 1 &  4 & 3 & 1 & 1 & 1  &    4.5s  &  TO     &  TO  & TO \\
        fetch2d corner                  & $\mathbb{Z}$ & 1 &  4 & 2 & 0 & 1 & 1  &    6.7s  &  TO     &  TO  & TO \\
        fetch2d corner 2p               & $\mathbb{Z}$ & 2 &  4 & 2 & 0 & 1 & 2  &   21.0s  &  TO     &  TO  & TO \\
        fetch2d corner 2p v2            & $\mathbb{Z}$ & 2 &  4 & 2 & 0 & 1 & 2  &   34.9s  &  TO     &  TO  & TO \\
        fetch2d corner 3p               & $\mathbb{Z}$ & 3 &  4 & 2 & 0 & 1 & 4  &   42.3s  &  TO     &  TO  & TO \\
        fetch3d stop signal             & $\mathbb{Z}$ & 1 &  6 & 4 & 1 & 1 & 1  &  142s    &  TO     &  TO  & TO \\
        fetch4d stop signal             & $\mathbb{Z}$ & 1 &  8 & 5 & 1 & 1 & --  &  TO      &  TO     &  TO  & TO \\
        distribution 2c 1m 1e            & $\mathbb{Z}$ & 2 &  3 & 2 & 0 & 1 & 3  &   4.2s  &  TO     &  TO  & TO \\
        distribution 2c 1m 2e            & $\mathbb{Z}$ & 2 &  3 & 3 & 0 & 1 & 3  &   6.0s  &  TO     &  TO  & TO \\
        distribution 3c 1m 1e            & $\mathbb{Z}$ & 3 &  4 & 2 & 0 & 1 & 4  &   12.8s &  TO     &  TO  & TO \\
        distribution 3c 1m 3e            & $\mathbb{Z}$ & 3 &  4 & 4 & 0 & 1 & 4  &   18.1s &  TO     &  TO  & TO \\
        distribution 3c 2m 1e            & $\mathbb{Z}$ & 3 &  5 & 3 & 0 & 1 & 4  &   15.4s &  TO     &  TO  & TO \\
        distribution 3c 2m 3e           & $\mathbb{Z}$  & 3 &  5 & 5 & 0 & 1 & 4  &   21.0s &  TO     &  TO  & TO \\
        distribution 4c 1m 1e            & $\mathbb{Z}$ & 4 &  5 & 2 & 0 & 1 & 5  &   31.7s &  TO     &  TO  & TO \\
        distribution 4c 1m 4e            & $\mathbb{Z}$ & 4 &  5 & 5 & 0 & 1 & 5  &   56.2s &  TO     &  TO  & TO \\
        distribution 4c 2m 1e            & $\mathbb{Z}$ & 4 &  6 & 3 & 0 & 1 & 5  &   40.0s &  TO     &  TO  & TO \\
        distribution 4c 2m 4e            & $\mathbb{Z}$ & 4 &  6 & 6 & 0 & 1 & 5  &   57.9s &  TO     &  TO  & TO \\
        \bottomrule
    \end{tabular}
\end{table}

There are five groups of benchmarks.
\begin{enumerate}
    \item The first group (\emph{intro example} through \emph{intro example r.\ v2}) consists of variations of the running example, which is a parameterized version of a benchmark presented in \cite{DBLP:conf/isola/MaderbacherWB24}. Version v2  uses only one justice assumption requiring that infinitely often there is no wind.
We also include  versions with real-valued state and input variables, for which the environment can pick any real-valued wind disturbance in $[0,1]$ and %
the loading and unloading zones are changed to the intervals $[min, min+1]$ and $[max-1,max]$.

\item
The second group consists of benchmarks designed to highlight some special cases.
The \emph{minimal example} reduces the number of initialization choices the environment has to make.
It has only one parameter and the initial value of the state variable is fixed.
Benchmark \emph{distinct strategies} shows how our algorithm performs when the strategy is qualitatively different based on the parameter value.
The task is to move between a base at $0$ and a objective at $p$.
In this example there are three types of strategies for $p=0$, $p>0$, and $p<0$. 
For $p=0$ the agent does not move at all, for $p>0$ it has to move right to get to the objective and left to get back to the base, for $p<0$ this is flipped.
Our algorithm handles this using the generalization of dissimilar expressions.
In the benchmark \emph{win by assume violation}, one guarantee cannot be satisfied, but the system can win by ensuring that the assumption becomes violated and remains so.

\item The third category (\emph{single gate} through \emph{gate wind}) contains benchmarks with a logistics tasks similar to the running example, except that there is no wind disturbance in most of them. Instead, access to the loading and unloading zones is secured by gates that are controlled by the environment   --
a gate can only be passed if it is open.
We include versions with a gate only at one end, with gates at either end, and with two lanes with independently controlled gates. The benchmark \emph{gate wind} combines a gate with wind disturbances.

\item The fourth category (\emph{fetch1d} through \emph{fetch4d stop signal}) includes benchmarks where the system has to collect tokens in a grid world.
Once the system has collected a token, the environment chooses the position of the next token.
The system has to always remain inside an arena defined by the parameters.
This resembles the rpgsolve benchmark \emph{robot-commute} \cite{DBLP:journals/pacmpl/HeimD24},  with the added arena constraint.
In the variants with a \emph{stop signal} the environment can temporarily stall the system.
Our benchmark \emph{fetch2d stop signal} resembles the sweap benchmark \emph{taxi-service} \cite{azzopardi2025} with an added arena constraint.
If not otherwise specified, the arena constraints are box-shaped (each coordinate is between 0 and the parameter). The \emph{fetch2d corner} benchmarks use an L-shaped arena where the direct path between two points might be unsafe and the system has to use different strategies based on the relative position of the agent and the next token.

\item The last category (\emph{distribution 2c 1m 1e} through \emph{distribution 4c 2m 4e}) consists of resource distribution benchmarks where the system has to allocate resources to multiple consumers in a manufacturing pipeline to ensure its continued operation.
Parameters are used to specify how many resources each consumer requires.
The environment controls the provided resources and the efficiency of the manufacturing process.
Different variants have different numbers of consumers that need to be supplied (c), have one or two kinds of material (m), and either a single efficiency input or one for each consumer (e).
For benchmarks with multiple materials each consumer needs to be supplied with the material associated with it.

\end{enumerate}

\toolname solves most  benchmarks.
For \emph{gate wind} our algorithm cannot construct a ranking function for a substrategy that sometimes moves away from the target.
For \emph{fetch4d stop signal} the state space is so large that the timeout is reached before the first concrete system is built.
Sweap can solve two benchmarks quickly, but it fails for all others.
No results are listed for the real valued benchmarks as these are not supported by sweap.
In the \emph{minimal example} the environment only initializes a single variable, in all other benchmarks it initializes at least two variables (parameters and states).
The temporal specification in \emph{win by assume violation} is very different from the other benchmarks and this structure seems to align well with the algorithms used by sweap.
Issy fails to synthesize a system for any of our benchmarks. For \emph{fetch1d} it shows realizability using uninterpreted function acceleration.

\section{Related Work}
\subsection{Parameterized Synthesis}
Parameterized synthesis was previously studied for models that differ significantly from ours. Distributed system, in particular, consist of an arbitrarily large set of communicating instances of a processes. The parameter is the number of instances.  The instances may be fully identical, chosen from a small set of systems, or they may differ in simple ways like an identifier. Typically, messages are simple and  consist of one token or a one of a small set of messages. Data does not usually play a major role. Nevertheless, for such parameterized systems, even verification is  unrealizable for all but the simplest models \cite{DBLP:series/synthesis/2015Bloem}. 

For certain classes, however, if the system is correct for a certain \emph{cutoff} parameter, then it is correct for any parameter value.  Jacobs et al.\ \cite{DBLP:conf/tacas/JacobsB12,DBLP:conf/vmcai/KhalimovJB13} propose a synthesis procedure for such systems in which they synthesize a system for the cutoff size, which yields a system that is correct for any parameter size. (More precisely, they synthesize one system that is correct for any parameter value up to the cutoff.) Unfortunately, even synthesis of such small distributed systems is typically undecidable because of the information forks they contain \cite{DBLP:conf/lics/FinkbeinerS05}. Thus, the authors suggest a  semi-decision procedure for such systems based on bounded synthesis \cite{DBLP:journals/sttt/FinkbeinerS13}.

Lazic et al.\ \cite{DBLP:conf/opodis/Lazic0WB17} similarly focus on parameterized distributed systems, but  consider fault-tolerant systems. Threshold-guarded fault-tolerant protocols count the number of messages of a certain type. If a process has received a certain threshold number of messages, it can draw conclusions about the state of the distributed system, even if some messages may have been dropped or inserted by a malicious participant. The method takes a sketch of the algorithm  (like in program sketching \cite{DBLP:journals/sttt/Solar-Lezama13}) and attempts to fill in the correct thresholds depending on the total number of processes or the number of unreliable processes.

Markgraf et al.\ \cite{DBLP:conf/aplas/MarkgrafHLNN20} also consider systems with a parameterized number of components and ``holes'' that need to be filled in by the synthesis tool. The authors reduce the parameterized synthesis problem for safety properties to the problem of solving a regular safety game. Solving this game gives a general solution. The authors show that L${}^*$ can be used to learn the winning region of this game, which is a regular language.

In contrast to these approaches, we consider systems that do not consist of multiple communicating instances of an identical process, but rather of one process with parameters that may, for instance, determine the minimal and maximal values of variables. Data plays a more important role in our setting, although certain types of communicating processes can be simulated.

Azeem et al. recently studied policy synthesis for parameterized Markov decision processes using decision tree learning \cite{DBLP:conf/vmcai/AzeemCKKMMW25}. This paper is similar in spirit to our work in that it solves small instances of a problem and then attempts to generalize the solution to the full set of instances. An important difference is that they consider reachability properties on Markov Decision Processes, while we consider GR(1) properties on non-probabilistic systems. Furthermore, the concepts of parameterization and generalization are different and the resulting Markov chains are not verified but only tested.  %

\subsection{Infinite-State Synthesis}

The problem of parameterized synthesis is closely related to infinite state reactive synthesis, which has gained popularity in recent years.
Most existing approaches fall into one of two broad categories. First, several approaches use an abstraction/refinement approach. They abstract the infinite state synthesis problem to a Boolean (finite) linear temporal logic formula, which is then handled by existing LTL synthesis tools. 
The second class solves the fixpoint problem with theories directly.

\paragraph{Boolean abstraction}
Temporal stream logic (TSL) \cite{finkbeiner2019a} is a logic that allows specification and synthesis of reactive systems that operate on data.
It separates the reactive control flow from the data processing.
The reactive program is synthesized by creating a Boolean abstraction in the form of an LTL formula which is solved by an of-the-shelf LTL synthesis tool.
While TSL in its basic form takes operations on data as uninterpreted functions, it has been extended to support theories.

In \cite{DBLP:conf/fmcad/MaderbacherB22} an abstraction/refinement procedure is introduced for TSL modulo theories. The procedure analyzes counter strategies to identify new predicates that violate the theory constraints. These predicates are used to improve the LTL abstraction.

A different approach for TSL modulo theories has been studied in \cite{ChoiFPS22}. 
It uses a SyGuS solver to find sequences of actions that can be use to make predicates from the specification true. 
The LTL formula is then augmented with ``lemmas'' that describe the effect of these action sequences, thus partially instantiating the theory. 
A  limitation of this approach is that synthesized programs can only use guard conditions that are already present in the specification.

The method presented in \cite{azzopardi2025} operates on infinite-state arenas, in which  moves are decided by Boolean variables selected by both system and environment.
This differs from the work on TSL, because both players are restricted to a finite number of actions.
In TSL only the system is restricted this way and the environment has infinitely many choices.
Synthesis is performed in a counter example guided abstraction refinement loop.
An LTL abstraction is created and given to an LTL synthesis tool. If the LTL synthesis procedure finds a counterexample, it is analyzed to see whether it corresponds to a concrete counterexample.
Unbounded strategy loops are handled by performing liveness refinement.
Counterexample lassos are analyzed for spurious non termination. If applicable, a assumption to enforce termination is added to the abstraction.

In \cite{rodriguez2025cav} a synthesis method for LTL modulo theories with infinite inputs and outputs and relations over values in consecutive time steps is shown.
It works via a counter example guided refinement loop that learns temporal tautologies. However, these tautologies only use next operators and can thus not be used to synthesize systems where the number of steps depends on an arbitrarily large state value or parameter.

Other approaches  uses LTL modulo theory specifications with  infinite inputs and outputs, but with a finite state space. It should be noted that this problem is decidable, in contrast to the TSL problems described so far.
Rodriguez et al.\ give a sound and complete realizability checking algorithm for this problem \cite{DBLP:conf/cav/RodriguezS23, DBLP:journals/jlap/RodriguezS24}.
The approach computes an equirealizable LTL formula that can be solved by an LTL synthesis tool.
The LTL abstraction is fully computed ahead of time and thus no online interaction with an LTL synthesis tool is required.
This method has been extended to an adaptive synthesis method \cite{DBLP:conf/aaai/Rodriguez024} where an SMT solver is used at runtime to map between concrete input/output values and the abstract values used by the synthesized program.
A further extension uses functional synthesis to obtain a deterministic program by precomputing the mapping \cite{DBLP:conf/atva/RodriguezGS24}. %

\paragraph{Fixpoints}
The reactive synthesis problem can be modeled as a two player game between a system and an environment. 
In the propositional case, these games can be solved using fixpoint algorithms \cite{DBLP:journals/apal/McNaughton93,DBLP:journals/tcs/Zielonka98}.
There  are several papers that extend these algorithms can be extended to work on state representations that consist of theory formulas.
This approach has been developed for safety specifications \cite{SamuelDK21}, for Linear Temporal Logic \cite{DBLP:conf/kbse/SamuelDK23}, and GR(1) \cite{DBLP:conf/isola/MaderbacherWB24} games.

Infinite state systems can require unbounded loops in which the number of iteration to reach the next goal is not constant, but instead depends on the environment inputs and can be arbitrary large.
Fixpoint acceleration \cite{DBLP:journals/pacmpl/HeimD24} was develop to handle this problem for fixpoint based synthesis algorithms. 
Instead of performing the fixpoint computation one time step at a time, fixpoint acceleration constructs an acceleration lemma that allows a single step in the fixpoint computation to cover arbitrary many time steps.
The method was first developed for reactive program games and has been extended to handle localized attractors \cite{DBLP:conf/cav/SchmuckHDN24} and TSL translation to reactive program games \cite{DBLP:journals/pacmpl/HeimD25}.
Issy \cite{DBLP:journals/corr/abs-2502-03013} is a synthesis tool that implements these algorithms and works on a specification language that combines games and temporal logic.

\paragraph{Other Methods}
The following methods do not fall cleanly into the Boolean abstraction or fixpoint classes. First, \cite{KatisFGGBGW18} presents a synthesis method inspired by IC3/PDR for safety specifications given as assume-guarantee contracts.

Deductive synthesis methods have been proposed by \cite{BeyeneCPR14, DBLP:conf/aaai/FaellaP23}.
ConSynth \cite{BeyeneCPR14} solves safety, reachability, and LTL games by encoding them into forall-exists Horn clauses.
It requires strategy templates written specifically for each specification.
Faella and Parlato \cite{DBLP:conf/aaai/FaellaP23} propose a method for solving safety games by encoding them into constraint Horn clauses.

Finally, Neider and Markgraf propose to use decision tree learning to learn the winning region of an infinite safety game, where the learner is driven by an oracle that can state that (1) an initial state must be in the winning region, (2) a bad state cannot be, or (3) if a given proponent (opponent) state is in the winning region, then so is one (all) of its successors. The algorithm is able to learn small trees that represent the winning region \cite{NeiderT16, NeiderM19}.

\section{Conclusions}

We have presented the first approach to synthesize parameterized systems that do not consist of multiple copies of the same system. Our synthesis approach uses specifications written in GR(1) that may include the parameters. Our approach is based on a counter-example driven synthesis loop. In the loop,  we use a novel idea of generating systems for different parameter values and generalizing them to a single parameterized system using anti-unification and syntax-guided synthesis. It is not easy to prove such systems correct using existing tools. Instead, we  generalize the invariants of the individual systems to one candidate parameterized invariant, and we use syntax-guided synthesis to build a candidate ranking function. We check whether the candidate proofs are correct and refine if not. A preliminary empirical evaluation shows that our approach outperforms existing tools on a set of new simple benchmarks that we have adapted to work with  tools for infinite state synthesis.

As we have shown, synthesis and verification are undecidable. Thus, we rely on heuristics. In particular, we use a generalization heuristic, and syntax-guided synthesis is constrained by the grammars we employ. We expect that these heuristics can be further improved to increase performance and applicability.

We only support a restricted form of parameterization. For instance, the number of variables is fixed and independent of the parameters.  Thus, our approach cannot handle a version of the running example in which the grid world has $p$ dimensions. Removing this limitation would be a interesting extension. Likewise, the specification can include parameters, but has a fixed structure. It can not include a subformula like $\bigwedge_{k\in [0,p]} \eventually(x = k)$.

We believe that our approach can be extended to synthesis of infinite state systems that are \emph{not} parameterized. One idea would be to automatically restrict the range of variables and then to generalize the resulting systems. Inspection of the systems in \cref{sec:evaluation} for which we introduced such parameters, indicates that they are correct even in the unbounded case. In general, however, it is not clear how to constrain the assumptions to preserve realizability.

Finally, it would be interesting to extend our algorithm to a parameterized version of LTL. We may, for instance, attempt to  generalize winning strategies for the parity games that correspond to an LTL formula.

\clearpage

\bibliographystyle{ACM-Reference-Format}
\bibliography{bibliography.bib}

\end{document}